%% file: pepm.tex
\documentclass[sigplan,screen]{acmart}

\AtBeginDocument{%
  }
\allowdisplaybreaks[1]
\usepackage{thmtools}
\usepackage{thm-restate}
\usepackage{graphicx}
\usepackage{bcprules}
\usepackage{xcolor}
\usepackage{titletoc}
\usepackage{stmaryrd}
\usepackage{listings,jvlisting}
\usepackage{multicol}
\newtheorem{remark}{Remark}
\usepackage{appendix}

\input{local}
\newif\ifdraft\draftfalse
\newif\iffull\fulltrue
\newif\iftwocol\twocoltrue
\ifdraft
\newcommand\nk[1]{\textcolor{red}{[#1 -nk]}}
\newcommand*\sk[1]{\textcolor{teal}{\scriptsize [#1 -sk]}}

\else
\newcommand\nk[1]{}
\newcommand*\sk[1]{}

\fi

\iffull
\newcommand*\Appendix[1]{Appendix~\ref{#1}}
\else
\newcommand*\Appendix[1]{the full version~\cite{TanakaPEPM24full}}
\fi

\iffull
\setcopyright{none}
\else
\setcopyright{acmlicensed}
\acmDOI{10.1145/3635800.3636965}
\acmYear{2024}
\copyrightyear{2024}
\acmSubmissionID{poplws24pepmmain-id5-p}
\acmISBN{979-8-4007-0487-1/24/01}
\acmConference[PEPM '24]{Proceedings of the 2024 ACM SIGPLAN International Workshop on Partial Evaluation and Program Manipulation}{January 16, 2024}{London, UK}
\acmBooktitle{Proceedings of the 2024 ACM SIGPLAN International Workshop on Partial Evaluation and Program Manipulation (PEPM '24), January 16, 2024, London, UK}
\received{2023-10-20}
\received[accepted]{2023-11-20}
\fi

\begin{document}

\title{Ownership Types for Verification of Programs with Pointer Arithmetic}

\author{Izumi Tanaka}
\orcid{0009-0006-3032-1595}
\affiliation{%
  \institution{The University of Tokyo}
  \city{}
  \country{Japan}
}
\email{tanaka-izumi@g.ecc.u-tokyo.ac.jp}

\author{Ken Sakayori}
\orcid{0000-0003-3238-9279}
\affiliation{%
  \institution{The University of Tokyo}
  \city{}
  \country{Japan}
}
\email{sakayori@is.s.u-tokyo.ac.jp}

\author{Naoki Kobayashi}
\orcid{0000-0002-0537-0604}
\affiliation{%
  \institution{The University of Tokyo}
  \city{}
  \country{Japan}
}
\email{koba@is.s.u-tokyo.ac.jp}

\renewcommand{\shortauthors}{Tanaka et al.}

\input{abs}

\keywords{Automated Program Verification, Ownership Types, Refinement Types, Pointer Arithmetic}

\maketitle

\input{intro}
\input{lang}
\input{type}

\input{infer}
\input{experiment}

\input{rel}
\input{concl}

\begin{acks}
  We thank Ryosuke Sato for discussions on this work.
  We would also like to thank anonymous referees for useful comments.
This work was supported by
JSPS KAKENHI Grant Number JP20H05703.
\end{acks}

\bibliographystyle{ACM-Reference-Format}
\bibliography{abbrv,koba,myref}

\iffull
\clearpage
\onecolumn
\appendix
\appendixpage
\startcontents[sections]
\printcontents[sections]{l}{1}{\setcounter{tocdepth}{2}}
\input{wf}
\input{soundness}
\clearpage
\input{benchmark}
\fi
\end{document}

%% file: local.tex
\newcommand\Z{\mathbb{Z}}
\newcommand\consort{\textsc{ConSORT}}
\newcommand\set[1]{\{#1\}}
\newcommand\Imp{\Rightarrow}
\newcommand\COL{\mathbin{:}}
\newcommand\INT{\mathbf{int}}
\newcommand\REF{\mathbf{ref}}
\newcommand\Tint[2]{\set{#1\COL\INT\mid #2}}
\newcommand\Tref[3]{(\lambda #1.#2)\;\mathbf{ref}^{#3}}

\newcommand\Tfunarg[1]{\langle #1\rangle}
\newcommand\TfunargL{\langle}
\newcommand\TfunargR{\rangle}
\newcommand\Tfunret[2]{\langle #1\mid #2\rangle}

\newcommand\p{\vdash}
\newcommand\FE{\Theta}
\newcommand\TE{\Gamma}
\newcommand\Empty{\mathit{Empty}}
\newcommand\emptyTE{\bullet}
\newcommand\ty{\tau}
\newcommand\To{\Longrightarrow}

\newcommand\form{\varphi}
\newcommand\dom{\mathit{dom}}
\newcommand\letexp[2]{\mathbf{let}\ #1=#2\ \mathbf{in}\ }
\newcommand\ifexp[2]{\mathbf{ifnp}\ #1\ \mathbf{then}\ #2\ \mathbf{else}\ }
\newcommand\OP{\mathtt{op}}
\newcommand\mkarray{\mathbf{alloc}}
\newcommand\mkref{\mathbf{ref}\;}
\newcommand\mkarrexp[2]{\mathbf{let}\ #1=\mkarray\ #2\ \mathbf{in}\ }
\newcommand\aliasexp[2]{\mathbf{alias}(#1=#2);}
\newcommand\ALIAS{\mathbf{alias}}
\newcommand\assertexp[1]{\mathbf{assert}(#1)}
\newcommand\pplus{\boxplus} %
\newcommand\red{\longrightarrow}

\newcommand\TRUE{\mathbf{true}}

\newcommand\sconfig[3]{\langle #1, #2, #3 \rangle}

\newcommand\redD{\longrightarrow_D}
\newcommand\update{\hookleftarrow}
\newcommand\ALIASFAIL{\mathbf{AliasFail}}
\newcommand\ASSERTFAIL{\mathbf{AssertFail}}
\newcommand*\FV{\mathbf{FV}}

\newcommand\funty{\sigma}
\newcommand\WF{\mathit{WF}}
\newcommand\pWF{\p_{\WF}}

\newcommand*\defeq{\stackrel{\text{def}}{=}}
\newcommand\Def{\overset{\text{def}}{\Longleftrightarrow}}
\newcommand\SAT[3]{\mathbf{SAT}(#1,#2,#3)}
\newcommand\SATv[4]{\mathbf{SATv}(#1,#2,\allowbreak #3,#4)}
\newcommand\Own[3]{\mathbf{Own}(#1,#2,#3)}
\newcommand\own[4][R]{\mathbf{own}(#2,#1,\allowbreak #3, #4)}

\newcommand\Addr[2]{(#1,#2)}
\newcommand\AddrArg[2]{#1,#2}
\newcommand\formTE[1]{\llbracket #1 \rrbracket}

\newcommand\seq[1]{\widetilde{#1}}

%% file: abs.tex
\begin{abstract}
  Toman et al. have proposed a type system for automatic verification of
  low-level programs,
  which combines ownership types and refinement types to enable
  strong updates of refinement types in the presence of pointer aliases.
  We extend their type system to support pointer arithmetic, and prove its soundness.
  Based on the proposed type system, we have implemented a prototype tool
  for automated verification of the lack of assertion errors of low-level programs
  with pointer arithmetic, and confirmed its effectiveness through experiments.
\end{abstract}

%% file: intro.tex
\section{Introduction}
\label{sec:intro}
Techniques for fully automated program verification\footnote{By \emph{fully} automated program verification,
we refer to automated verification of functional correctness (typically expressed by using assertions)
of programs that require no or little user intervention. It should be contrasted with
semi-automated verification based on
the verification condition generation approach~\cite{DBLP:conf/lpar/Leino10,DBLP:conf/esop/FilliatreP13,DBLP:series/natosec/0001SS17}, which require annotations of loop invariants and
pre/post conditions of recursive procedures.}
have made a lot of progress in recent years~\cite{SLAM,KSU11PLDI,Bjorner15,DBLP:journals/fmsd/KomuravelliGC16,DBLP:conf/cav/GurfinkelKKN15,DBLP:conf/cav/KahsaiRSS16,DBLP:conf/esop/TomanSSI020,DBLP:journals/toplas/MatsushitaTK21}, but
it is still challenging to automatically verify low-level imperative programs that involve mutable references.
Due to possible aliasing, it is in general difficult to keep track of information about the value
stored in each reference cell in a sound, precise, and efficient manner.
Toman et al.~\cite{DBLP:conf/esop/TomanSSI020} addressed this problem
by combining refinement types with fractional ownerships, and constructed an automated verification tool
called \consort{}
for a simple imperative language with mutable references and first-order recursive functions.
In the present paper, we extend their type system to support pointer arithmetic.

\paragraph{Review of the Ideas of \consort{}}
Before explaining our extension, let us first explain the ideas of \consort{}.
Consider the program in Figure~\ref{fig:consort1}.
The first line creates a reference cell whose value is initialized to \(0\) and binds \(x\) to it,
the second line updates the value of \(x\) to \(1\), and the third line asserts that the value of
\(x\) is indeed \(1\). The comments of the first and second lines (after ``//'') show
the (refinement) type environments after each line. Here, please note that the type of \(x\)
is also updated: on the first line, the type of \(x\) indicates that the value stored in \(x\) is \(0\),
while the type of \(x\) on the second line  indicates that the value stored in \(x\) is now \(1\).
Please ignore the superscript \(1\) of \(\mathbf{ref}\) for the moment.
Based on that, we can conclude that the assertion on the third line succeeds.

\begin{figure}
  \begin{align*}
&    \letexp{x}{\mkref{0}} && \mbox{// }x\COL \Tint{\nu}{\nu=0}\ \mathbf{ref}^1\\
&    x := 1; && \mbox{// }x\COL \Tint{\nu}{\nu=1}\ \mathbf{ref}^1\\
&    assert(*x = 1); && \mbox{// ok}
  \end{align*}
  \caption{A \consort{} Program}
  \label{fig:consort1}
\end{figure}

As the example above shows, we need to strongly update %
refinement types in the presence of mutable
reference. A naive update of refinement types leads to a wrong reasoning.
For example, consider the program in Figure~\ref{fig:consort2}.
Since the variable \(y\) is an alias of \(x\), the value of \(y\) (i.e., \(*y\))
is updated to \(1\) on the third line, and thus the assertion on the
last line should fail. The type of \(y\) in the comment on the third line
has not been updated, so a naive type system may wrongly judge that
the program is safe. One way to correctly maintain refinement types
would be to keep track of alias information; in the case of this example,
if a type system keeps information that \(y\) is an alias of \(x\), then
it can update the types of \(x\) and \(y\) together on the third line.
It would, however, be difficult in general,
as complete alias information is not always available statically.

\begin{figure*}
  \begin{align*}
&    \letexp{x}{\mkref{0}} && \mbox{// }x\COL \Tint{\nu}{\nu=0}\ \mathbf{ref}^1\\
&    \letexp{y}{x} && \mbox{// }x\COL \Tint{\nu}{\nu=0}\ \mathbf{ref}^1, y\COL\Tint{\nu}{\nu=0}\ \mathbf{ref}^0\\
&    x := 1; && \mbox{// }x\COL \Tint{\nu}{\nu=1}\ \mathbf{ref}^1, y\COL\Tint{\nu}{\nu=0}\ \mathbf{ref}^0\\
&    assert(*y = 0); && \mbox{// wrong}
  \end{align*}
  \caption{A Buggy \consort{} Program}
  \label{fig:consort2}
\end{figure*}
\begin{figure*}[tbh]
  \begin{align*}
&    \letexp{x}{\mkref{0}} && \mbox{// }x\COL \Tint{\nu}{\nu=0}\ \mathbf{ref}^1\\
&    \letexp{y}{x} && \mbox{// }x\COL \Tint{\nu}{\nu=0}\ \mathbf{ref}^1, y\COL\Tint{\nu}{\nu=0}\ \mathbf{ref}^0\\
&    x := 1; && \mbox{// }x\COL \Tint{\nu}{\nu=1}\ \mathbf{ref}^1, y\COL\Tint{\nu}{\nu=0}\ \mathbf{ref}^0\\
    &    \aliasexp{x}{y} && \mbox{// }x\COL \Tint{\nu}{\nu=1}\ \mathbf{ref}^{0.5},
      y\COL\Tint{\nu}{\nu=1}\ \mathbf{ref}^{0.5}\\
&    \assertexp{*y = 1}; && \mbox{// ok}
  \end{align*}
  \caption{A Correct \consort{} Program}
  \label{fig:consort3}
\end{figure*}

To address the issue above, \consort{} combines fractional
ownerships with refinement types.
The type of a pointer to a value of type \(\tau\)
is of the form \(\tau\;\mathbf{ref}^o\), where \(o\), called an ownership,
ranges over \([0,1]\). The type system of \consort{} enforces that
(i) if there are multiple aliases to a reference cell,
the sum of their ownerships is no greater than \(1\),
(ii) a pointer can be used for update only when its ownership is \(1\)
(so that even if other aliases exist, their ownerships are \(0\)),
and (iii) when a value is read through a pointer with ownership \(0\),
information carried by the refinement type is not trusted.
In the case of the program in Figure~\ref{fig:consort2},
the types indicated in the comments follow the constraints described above.
On line 2, the sum of ownerships is kept to \(1+0=1\);
on line 3, \(x\) can be updated because the ownership of \(x\) is \(1\);
on line 4, as the type of \(y\) is \(\Tint{\nu}{\nu=0}\ \mathbf{ref}^0\),
the refinement type \(\Tint{\nu}{\nu=0}\) of \(*y\) is not trusted;
hence, the wrong conclusion \(*y=0\) is avoided.

Figure~\ref{fig:consort3} shows a variation of the program in
Figure~\ref{fig:consort2}, where the assertion on the last line has
been replaced with \(*y=1\).
The expression \(\aliasexp{x}{y}\) assumes that \(x\) and \(y\) are
aliases; it is provided as a hint to allow the type system
to merge the ownerships and type information of \(x\) and \(y\)
and redistribute them between the types of \(x\) and \(y\).
In this case, 0.5 ownership is redistributed to \(x\) and \(y\),
and the refinement type information \(\Tint{\nu}{\nu=1}\) is propagated to
\(y\) (\(\Tint{\nu}{\nu=0}\) is discarded, as the refinement type
with \(0\) ownership cannot be trusted, as described above).
Based on the type information, we can conclude that the assertion on
the last line is correct. The \(\ALIAS\) expressions can be
automatically inserted by static must-alias analysis, or
provided by a programmer as hints.

\paragraph{Our Proposal in This Paper.}
We extend the type system of \consort{} to support pointer arithmetic (i.e.,
pointer addition and subtraction; other operations such as XOR operations and casting from/to integers
are out of scope).
We also generalize refinement types of arrays (or a chunk of memory
consisting of consecutive memory addresses); \consort{} supported only
integer arrays (not arrays of pointers), and
did not allow an array to be split into multiple arrays.

\newcommand\INIT{\mathit{init}}
\begin{figure*}[tbh]
  \begin{align*}
      &\INIT(x, p) \{\ 
    \mbox{// } \INIT\COL
    \Tfunarg{x\COL\INT, p\COL \Tref{i}{\Tint{\nu}{\TRUE}}{\set{[0,x-1]\mapsto 1}}} 
      \to  \Tfunret{x\COL\INT, p\COL \Tref{i}{\Tint{\nu}{\nu=0}}{\set{[0,x-1]\mapsto 1}}}{\INT}\\
    & \quad
      \begin{aligned}
       &\ifexp{x}{1} &&\\
       & p := 0;
         &&\mbox{// } p\COL \Tref{i}{\Tint{\nu}{i=0\Imp\nu=0}}{\set{[0,x-1]\mapsto 1}}\\
       &\letexp{q}{p+1}
        &&\mbox{// } p\COL \Tref{i}{\Tint{\nu}{\nu=0}}{\set{[0,0]\mapsto 1}},
                          q\COL \Tref{i}{\Tint{\nu}{\TRUE}}{\set{[0,x-2]\mapsto 1}}\\
       & \INIT(x-1,q);
        &&\mbox{// } p\COL \Tref{i}{\Tint{\nu}{\nu=0}}{\set{[0,0]\mapsto 1}}, q\COL \Tref{i}{\Tint{\nu}{\nu=0}}{\set{[0,x-2]\mapsto 1}}\\
       & \aliasexp{q}{p+1}
        &&\mbox{// } p\COL \Tref{i}{\Tint{\nu}{\nu=0}}{\set{[0,x-1]\mapsto 1}}, q\COL \Tref{i}{\Tint{\nu}{\TRUE}}{\emptyset}\\
       & 0\quad \} &&
      \end{aligned}
  \end{align*}
  \caption{A Motivating Example}
  \label{fig:example1}
\end{figure*}

As a motivating example, consider the program in Figure~\ref{fig:example1},
which defines a recursive function \(\INIT\).
The function \(\INIT\) takes an integer \(n\) and a pointer \(p\)
as arguments, and initializes consecutive \(n\) memory cells
from \(p\) (i.e., \(p[0],\ldots,p[n-1]\)) to \(0\), by
first setting \(p[0]\) to \(0\), and then by calling \(\INIT\) recursively for \(p+1\).
We generalize the type of a pointer so that a pointer may have an ownership
for multiple memory cells. The extended pointer type is of the form
\(\Tref{i}{\tau}{r}\), where the ownership \(r\) is now a \emph{function}
from the set of relative addresses to the set of rational numbers in \([0,1]\).
For example, if a pointer \(p\) has
the ownership \(\set{[0,n-1]\mapsto 1}\) (which represents the function
\(r\) such that \(r(i)=1\) if \(0\le i\le n-1\) and \(r(i)=0\) otherwise),
then it has the full ownership to read/update
 \(p[0],\ldots,p[n-1]\).
The part \(\tau\) of \(\Tref{i}{\tau}{r}\), which may depend on \(i\),
describes the type of the value stored in \(p[i]\).
For example,
\(\Tref{i}{\Tint{\nu}{i=0\Imp\nu=0}}{\set{[0,n-1]\mapsto 1}}\)
describes the type of a pointer to an integer array of size \(n\),
whose first element is \(0\).

In Figure~\ref{fig:example1}, the type of the function \(\INIT\) on
the first two lines describes that the function takes
an integer \(x\) and a pointer \(p\) with the full ownerships for
\(p[0],\ldots,p[x-1]\), 
and returns an integer in a state where the values of \(
p[0],\ldots,p[x-1]\) have been updated to \(0\).\footnote{
Since there are no ownerships for \(i\not\in [0,x-1]\),
there is an implicit assumption
\(0\le i\le x-1\) for the refinement condition \(\nu=0\).}
Here the type on the right-hand side of the arrow not only describes the return type of the function, namely \( \INT \), but also the types of the arguments after the function call.
On the fourth line, \( \mathbf{ifnp} \) checks if \( x \) is ``not positive'', and if that is the case the function returns status 1.
On the fifth line, the refinement condition is updated to
\(i=0\Imp \nu=0\), since the value of (the cell pointed to by) \(p\) has been updated.
On the next line, a new pointer \(q\) is created and the ownerships for
\(p[1],\ldots,p[x-1]\) have been passed to \(q\).
On the seventh line, \(\INIT\) is recursively called
and the type of \(q\) is updated.
On the eighth line, the types of \(p\) and \(q\) are merged into the type of \(p\).
The alias command \(\ALIAS(q=p+1)\) can be automatically inserted in this case
by a simple must-alias analysis.

We have formalized a new type system based on the ideas described above.
Based on the type system, we have also implemented a prototype
tool for automated verification of programs with pointer arithmetic.
The tool takes simply typed programs as inputs and infers the ownership functions and the refinement predicates.
If the inference succeeds, it implies that no assertion failures will occur by the soundness of our type system.
Note that the comments in Figure~\ref{fig:example1} were shown merely for the sake of explanation, and these are not manual annotations.

The rest of this paper is organized as follows.
Section~\ref{sec:lang} introduces the target language.
Section~\ref{sec:type} introduces our type system and proves its soundness.
Section~\ref{sec:infer} describes a type inference procedure 
and our prototype tool.
Section~\ref{sec:experiment} reports experimental results.
Section~\ref{sec:related} discusses related work and
Section~\ref{sec:conc} concludes the paper.

%% file: lang.tex
\section{Target Language}
\label{sec:lang}

This section defines the syntax and operational semantics
of the target language of our verification. It is a simple imperative language with
pointer arithmetic.
\subsection{Syntax}

The sets of \emph{definitions}, \emph{expressions}, and \emph{programs},
ranged over by \(d\), \(e\), and \(P\) respectively, are defined by:
\begin{align*}
	d&::=f\mapsto(x_1,\dots,x_n)e\\
	e&::=n\mid x\mid\letexp{x}{e_1}e_2\\
	 &\mid\ \ifexp{x}{e_1}{e_2}\mid\letexp{x}{f(y_1,\dots ,y_n)}e\\
	 &\mid\ \mkarrexp{x}{n}{e}\mid\letexp{x}{*y}{e} \mid x:=y;e \\
   &\mid\letexp{x}{y\pplus z}{e} \mid \letexp{x}{y\ \OP\ z}e\\
	 &\mid\ \aliasexp{x}{*y}e\mid\aliasexp{x}{y\pplus z}e\\
	 &\mid\ \assertexp{\form};e \\
  P&::=\langle\set{d_1,\dots,d_n},e\rangle.
\end{align*}
\noindent
Here, the meta-variables \(x, y, z, ...\) denote variables, and \(f\) denotes a function name.
The meta-variable \(\form\) denotes a formula of the first-order 
logic over integer arithmetic.
The meta-variable \(\OP\) ranges over a set of operators on integers such as \(+\) and \(-\).

We informally explain the meaning of expressions; the formal semantics is given
in Section~\ref{sec:os}.
The expression \(\ifexp{x}{e_1}{e_2}\) reduces to \(e_1\) if \(x\le 0\), %
and reduces to \(e_2\) otherwise.
The expression \(\mkarrexp{x}{n}{e}\) allocates a memory region of
size \(n\) (where the size denotes the number of words, rather than bytes),
binds \(x\) to the pointer to it, and evaluates \(e\).
The term \(*y\) denotes a dereference of the pointer \(y\),
and the term \(y\pplus z\) denotes a pointer arithmetic; we use \(\pplus\) instead
of \(+\) to avoid the confusion with the integer operator \(+\).
The alias expressions 
\(\aliasexp{x}{*y}e\) and \(\aliasexp{x}{y\pplus z}e\)
check whether \(x=*y\) and \(x=y\pplus z\) hold, respectively, and proceed to evaluate \(e\) if that is the case;
otherwise they stop the execution of the whole program.
The assert expression \(\assertexp{\form};e\) checks whether \(\form\) is true,
and if so, proceeds to evaluate \(e\); otherwise, the program is aborted with an
(assertion) error.
The purpose of our type-based verification is to check that no assertion error occurs.

Compared with the language of \consort{}~\cite{DBLP:conf/esop/TomanSSI020},
we have added the pointer addition \(y\pplus z\), and replaced the primitive
\(\letexp{x}{\mkref{y}}{e}\) for creating a single reference cell
with the primitive \(\mkarrexp{x}{n}{e}\) for creating a contiguous memory region of arbitrary size.
While this may seem to be a small extension, it imposes a significant challenge
in the design of the type system and type inference discussed in the following sections.
In particular, as discussed in Section~\ref{sec:intro},
the ownership of a pointer needs to be extended from a single number (representing
the ownership of the single cell referred to by the pointer) to a function that maps each
relative address to a number (representing the ownership for the cell referred to by the relative address).

\begin{example}
  \label{ex:init}
  Recall the motivating example in Figure~\ref{fig:example1}.
  It is expressed as the function definition \(d = \INIT\mapsto (x,p)e_1\), where
  \(e_1\) is:
  \begin{align*}
    &\ifexp{x}{1}   \\&p:=0; %
    \letexp{j}{1} %
    \letexp{q}{p\pplus j} %
    \letexp{y}{x-j} \\    &
    \letexp{z}{\INIT(y,q)} \aliasexp{q}{p\pplus j}0
  \end{align*}
  Let \(P = (\set{d},e_2)\) where \(e_2\) is
  \begin{align*}
&    \mkarrexp{a}{3}\letexp{u}{\INIT(3,a)}\\
&    \letexp{k}{2}\letexp{q}{a\pplus k}\letexp{x}{*q}\assertexp{x=0};0.
  \end{align*}
  Then the assertion \(\assertexp{x=0}\) should succeed.
  If we replace \(\INIT(3,a)\) with \(\INIT(2,a)\), then the assertion may fail.
    \qed
\end{example}

\begin{remark}
  As explained in Section~\ref{sec:intro}, alias expressions are a kind
  of program annotations to allow the type system to redistribute ownership
  and refinement type information. Following the work of \consort{},
  to keep the type system simple,
  we assume that programs are annotated with proper alias expressions in Sections~\ref{sec:lang}--\ref{sec:infer}.
  It is not difficult to insert those annotations automatically by using  a  must-alias analysis.
  According to our inspection of the \consort{} benchmark\footnote{\url{https://github.com/SoftwareFoundationGroupAtKyotoU/consort.}}
  and our own experiments, a straightforward, extremely naive local must-alias analysis (such as the one
  that just inserts \(\ALIAS({p}={q\pplus n})\) in the body of \(\letexp{p}{q\pplus n}{\cdots}\)) usually suffices.
  In fact, as described in Section~\ref{sec:experiment}, alias expressions are automatically inserted in our prototype tool,
  so that programmers need not provide alias annotations. \qed
\end{remark}

\begin{remark}
  We consider only pointer addition (and subtraction through the addition of a negative integer) as pointer operations
  for the sake of simplicity. It would not be difficult to extend our method to support pointers to structures.
  It would, however, be difficult to deal with the XOR operation supported
  in some of the semi-automated verification methods~\cite{DBLP:conf/pldi/SammlerLKMD021}. \qed
\end{remark}
\subsection{Operational Semantics}
\label{sec:os}
This section  defines the operational semantics of our language.
A \emph{configuration} is a triple \(\sconfig{R}{H}{e}\).
Here,
\(R\), called a \emph{register file}, is a map from a finite set of variables
to the set \(V\) consisting of integers and pointer addresses,
where a pointer address is a pair \(\Addr{a}{i}\) of integers;
\(a\) denotes the absolute address of a base pointer
and \(i\) is the relative address.
A \emph{heap} \(H\) is a map from a finite set of
pointer addresses to \(V\).
The reduction relation \(\redD\) for configurations is given 
in Figure~\ref{fig:lang}.
In the figures, \(E\) ranges over the set of evaluation contexts, defined by:
\(E::=[\,]\mid\letexp{x}{E}e\).

As most of the rules should be self-explanatory, we only explain a few rules.
The rule \rn{R-MkArray} is for allocating a memory region.
It picks a fresh address \(a\) (which is represented as a label,
rather than an integer; thus, memory regions with different base pointers never overlap),
and allocates a region consisting of
the addresses from \(\Addr{a}{0}\) to \(\Addr{a}{n-1}\).
Initial values stored in those memory cells can be arbitrary.
In the rules \rn{R-AliasDeref} and \rn{R-AliasDerefFail},
it is checked whether \(x=*y\) holds. If that does not hold,
the configuration is reduced to \(\ALIASFAIL\). Note
that we distinguish it from \(\ASSERTFAIL\) in the rule
\rn{R-AssertFail}. The goal of our verification is to check that
\(\ASSERTFAIL\) does not occur, assuming that 
the alias annotations are correct (i.e., \(\ALIASFAIL\) never occurs).
In the rule \rn{R-Assert}, \([R]\form\) represents the formula obtained
by replacing each variable \(x\in\dom(R)\) in \(\form\) with \(R(x)\).

\input{fig-lang}

\begin{remark}
  A run-time configuration in \consort{}~\cite{DBLP:conf/esop/TomanSSI020}
  was expressed as a quadruple consisting of a register \(R\), a heap \(H\), a stack frame, and an expression \(e\).
  We have removed the stack frame for the sake of simplicity and revised the rules accordingly. \qed
\end{remark}

%% file: fig-lang.tex
\begin{figure*}
\begin{multicols}{2}
\infrule[R-Context]
    {\sconfig{R}{H}{e} \redD \sconfig{R'}{H'}{e'}}
    {\sconfig{R}{H}{E[e]} \redD \sconfig{R'}{H'}{E[e']}}

\infrule[R-LetInt]
    {x'\notin\dom(R)}
    {\sconfig{R}{H}{\letexp{x}{n}e} \\\redD \sconfig{R\set{x'\mapsto n}}{H}{[x'/x]e}}

\infrule[R-LetVar]
    {x'\notin\dom(R)}
    {\sconfig{R}{H}{\letexp{x}{y}e} \\\redD \sconfig{R\set{x'\mapsto R(y)}}{H}{[x'/x]e}}

\infrule[R-IfTrue]
    {R(x)\le 0}
    {\sconfig{R}{H}{\ifexp{x}{e_1}{e_2}}\\ \redD \sconfig{R}{H}{e_1}}

\infrule[R-IfFalse]
    {R(x)>0}
    {\sconfig{R}{H}{\ifexp{x}{e_1}{e_2}} \\\redD \sconfig{R}{H}{e_2}}

\infrule[R-Deref]
        {%
          x'\notin\dom(R)}
    {\sconfig{R}{H}{\letexp{x}{*y}e} \\\redD \sconfig{R\set{x'\mapsto
          H(R(y))}}{H}{[x'/x]e}}

\infrule[R-Assign]
    {R(x)=\Addr{a}{i}\andalso \Addr{a}{i}\in\dom(H)}
    {\sconfig{R}{H}{x:=y;e} \\\redD \sconfig{R}{H\set{\Addr{a}{i}\update R(y)}}{e}}

\infrule[R-AddPtr]
    {R(y)=\Addr{a}{i}\andalso  x'\notin\dom(R)}
    {\sconfig{R}{H}{\letexp{x}{y\pplus z}e}\\ \redD \sconfig{R\set{x'\mapsto \Addr{a}{i+R(z)}}}{H}{[x'/x]e}}

    \infrule[R-MkArray]
        {\Addr{a}{0}\notin\dom(H)\andalso x'\notin\dom(R)\\
        H'=H\set{\Addr{a}{0}\mapsto m_0}\cdots\set{\Addr{a}{n-1}\mapsto m_{n-1}}}
    {\sconfig{R}{H}{\letexp{x}{\mkarray\;n}e} \redD\\
     \sconfig{R\set{x'\mapsto \Addr{a}{0}}}{H'}{[x'/x]e}}

\infrule[R-Call]
    {f\mapsto (x_1,\dots ,x_n)e\in D}
    {\sconfig{R}{H}{\letexp{x}{f(y_1,\dots ,y_n)}e'} \\\redD \sconfig{R}{H}{\letexp{x}{[y_1/x_1]\cdots[y_n/x_n]e}e'}}

\infrule[R-AliasDeref]
    {H(R(y))=R(x)}
    {\sconfig{R}{H}{\aliasexp{x}{*y}e} \redD \sconfig{R}{H}{e}}

\infrule[R-AliasAddPtr]
    {R(x)=\Addr{a}{i+R(z)}\andalso R(y)=\Addr{a}{i}}
    {\sconfig{R}{H}{\aliasexp{x}{y\pplus z}e} \redD \sconfig{R}{H}{e}}

\infrule[R-AliasDerefFail]
    {H(R(y))\ne R(x)}
    {\sconfig{R}{H}{\aliasexp{x}{*y}e} \redD \ALIASFAIL}

\infrule[R-AliasAddPtrFail]
    {R(x)=\Addr{a_1}{i_1}\andalso R(y)=\Addr{a_2}{i_2}\\ a_1\not=a_2\lor i_1\not=i_2+R(z)}
    {\sconfig{R}{H}{\aliasexp{x}{y\pplus z}e} \redD \ALIASFAIL}

\infrule[R-Assert]
    {\models [R]\form}
    {\sconfig{R}{H}{\assertexp{\form};e} \redD \sconfig{R}{H}{e}}

\infrule[R-AssertFail]
    {\not\models [R]\form}
    {\sconfig{R}{H}{\assertexp{\form};e} \redD \ASSERTFAIL}
\end{multicols}
    \caption{Reduction Rules}
    \label{fig:lang}

\end{figure*}

%% file: type.tex
\section{Type System}
\label{sec:type}
This section introduces our ownership refinement type system.

\subsection{Types}

The syntax of types (ranged over by \( \ty \)) and function types (ranged over by \( \funty \) ) is given by:
\begin{align*}
    \ty&::=\Tint{\nu}{\form}\mid \Tref{i}{\ty}{r}\\
    \funty&::=\langle x_1:\ty_1, \dots ,x_n:\ty_n\rangle\rightarrow\langle x_1:\ty_1',\dots ,x_n:\ty_n'\mid \ty\rangle
\end{align*}

\noindent
Here, 
the meta-variable \(\form\) ranges over the set of formulas of first-order
logic with integer arithmetic and \(r\) ranges over
the set of functions (called \emph{ownership functions})
from integers to rational numbers in the range of
\([0,1]\).

The type \(\Tint{\nu}{\form}\) describes
an integer \(\nu\) that satisfies \(\form\);
we often call \(\form\) a \emph{refinement predicate}. For example,
\(\Tint{\nu}{\nu>0}\) is the type of positive integers.
We sometimes just write \(\INT\) for \(\Tint{\nu}{\top}\),
where \(\top\) denotes \(\TRUE\).
The type \(\Tref{i}{\ty}{r}\) describes a pointer \(p\) to a memory region, 
where the pointer has \(r(i)\) ownership for the address \(p+i\) and
\(\ty\) (which may depend on \(i\)) describes the type of the value stored at
\(p+i\). The ownership \(r(i)\) ranges over \([0,1]\).
If \(r(i)=1\), the cell at \(p+i\) can be read/written through
the pointer. If \(0<r(i)<1\), then the cell at \(p+i\) can only be read;
and if \(r(i)=0\), then the cell at \(p+i\) cannot be accessed at all through the
pointer \(p\).
For example, \(\Tref{i}{\Tint{\nu}{0\le i\le 9 \Imp \nu=i}}{r}\)
with \(r(i)=1\) for \(i\in\set{0,\ldots,9}\) and \(r(i)=0\) for
\(i\in \Z\setminus\set{0,\ldots,9}\) describes (a pointer to the first address of)
a mutable array of size \(10\) whose \(i\)-th value is \(i\).

We often write \(\set{[m,n]\mapsto u}\) for the ownership function \(r\)
such that \(r(i)=u\) if \(m\le i\le n\) and \(r(i)=0\) otherwise.
In this section, we allow an ownership function to be
any mathematical function in \(\Z\to [0,1]\), but we will restrict
the shape of ownership functions to those of the form
\(\set{[m,n]\mapsto u}\) in Section~\ref{sec:infer} for
  the purpose of type inference.

  A function type is of
  the form \(\langle x_1:\ty_1, \dots ,x_n:\ty_n\rangle\rightarrow\langle x_1:\ty_1',\dots ,x_n:\ty_n'\mid \ty\rangle\),
  where \(\ty_i\) and \(\ty'_i\) are the types of the \(i\)-th argument
  respectively \emph{before} and \emph{after} a function call, and
  \(\ty\) is the type of the return value.
  For example, as given in Section~\ref{sec:intro}, the type:
  \begin{align*}
&      \Tfunarg{x\COL\INT, p\COL \Tref{i}{\INT}{\set{[0,x-1]\mapsto 1}}}\\&
    \to 
    \Tfunret{x\COL\INT, p\COL \Tref{i}{\Tint{\nu}{\nu=0}}{\set{[0,x-1]\mapsto 1}}}
            {\INT}
  \end{align*}
  describes a function that takes an integer \(x\) and
  (a pointer to) an array of size \(x\) with the full ownership,
  sets all the elements to \(0\), and returns an integer.

\begin{remark}
  Here we summarize the difference from the types of \consort{}~\cite{DBLP:conf/esop/TomanSSI020}.
  In \consort{}, the type of a pointer is of the form \(\tau\; \mathbf{ref}^r\) where
  \(\tau\) is the type of a value stored in the target cell of the pointer, \(r\in[0,1]\) is
  the ownership for the cell. To allow pointer arithmetic, we needed to extend it so that both
  the \(\tau\)-part (now described as \(\lambda i.\tau\)) and \(r\)-part carry information about the memory region around
  the cell, not just the cell. \consort{} also supports an array type of the form \((\lambda i.\tau)\;\mathbf{array}^r\)
  (though not formalized in the paper~\cite{DBLP:conf/esop/TomanSSI020})
  but the ownership \(r\in[0,1]\) is a single value that expresses the ownership for the entire array, and the decomposition
  of the ownership for subarrays (which is crucial for the support of pointer arithmetic) is not allowed.
  \consort{} supports context polymorphism; we have omitted it for the sake of simplicity, as it is orthogonal to our extension
  for supporting pointer arithmetic. \qed
\end{remark}
  
\subsection{Typing Rules}
A type judgment for expressions is
of the form \(\FE\mid\TE\p e:\ty\To\TE'\), where \(\FE\), called
a function type environment, is a map from variables to function types,
and \(\TE\), called a type environment, is a sequence of type bindings of
the form \(x\COL\ty\). We sometimes view the type environment as a map from variables
to types, and implicitly allow re-ordering of type bindings (subject to
the well-formedness condition below). We write \(\dom(\TE)\) for the domain of the map \(\TE\).
The judgment
\(\FE\mid\TE\p e:\ty\To\TE'\) means that the expression \(e\) is well-typed under \(\FE\) and \(\TE\), 
and evaluates to a value of the type \(\ty\), modifying \(\TE\) to \(\TE'\).
We implicitly require the following well-formedness conditions on types, type environments,
and type judgments. 
First, the variables that occur in refinement predicates and ownership functions
must be bound as integer variables. For example,
the type \(\Tref{i}{\Tint{\nu}{\nu>x}}{\set{[0,x-1]\mapsto 1}}\)
is well formed under the type environment \(x\COL\INT\) but
not under the empty type environment \(\emptyTE\).
Second, for each pointer type \(\Tref{i}{\ty}{r}\),
we require that \(r(i)=0\) implies \(\Empty(\ty)\), where 
\(\Empty\) is inductively defined by:
(i) \(\Empty(\Tint{\nu}{\form})\) if \(\form\) is equivalent to \(\TRUE\),
and (ii) \(\Empty(\Tref{i}{\ty}{r})\) if
\(\Empty(\ty)\) and \(r(i)=0\) for every \(i\in\Z\).
We write \(\TE\pWF \ty\) if \(\ty\) is well-formed under \(\TE\);
the formal definition is given in \Appendix{sec:wf}.

The typing rules for expressions and programs and auxiliary rules for
subtyping are given in Figures~\ref{fig:type1}--\ref{fig:type3}.
In the rules,
\(\TE[x:\ty]\) represents a type environment \(\TE\)
such that \(\TE(x)=\ty\).
We write
\(\TE, x:\ty\) for the extension of \(\TE\) (where \(x\not\in\dom(\TE)\)) with the binding \(x:\ty\), 
and \(\TE[x\update \ty]\) for the type environment obtained from
\(\TE\) by updating the binding of \(x\) to \(\ty\).
The empty environment is represented as \(\emptyTE\).

\input{fig-type}

We explain key rules and notations used in the rules below.
Except for the rules for new primitives, i.e.~\rn{T-MkArray}, \rn{T-AddPtr}, and \rn{T-AliasAddPtr},
the typing rules have been made to look superficially similar to those of \consort{}~\cite{DBLP:conf/esop/TomanSSI020}
by appropriately extending operations on types and ownerships.
The rule \textsc{T-Var} is for variables. 
The operation \(\ty_1+\ty_2\)
 is inductively defined by:
\begin{align*}
  \Tint{\nu}{\form_1} + \Tint{\nu}{\form_2} &= \Tint{\nu}{\form_1\land\form_2} \\
  \Tref{i}{\ty_1}{r_1} + \Tref{i}{\ty_2}{r_2} &= \Tref{i}{\ty_1+\ty_2}{r_1+r_2}
\end{align*}
Here, \(r_1+r_2\) denotes the pointwise addition of the ownerships.
Intuitively, \rn{T-Var} splits the current type \(\ty_1+\ty_2\) of \(x\)
into \(\ty_1\) and \(\ty_2\), respectively
for the expression \(x\) and the type environment
after the evaluation of the expression.
For example,
\begin{align*}
&  \FE\mid \TE
\p x:\Tref{i}{\Tint{\nu}{\form_1}}{\set{0\mapsto 1}}\\
&\quad\To
x:\Tref{i}{\Tint{\nu}{\form_2}}{\set{1\mapsto 1}}
\end{align*}
and
\begin{align*}
&  \FE\mid \TE
  \p x:\Tref{i}{\Tint{\nu}{\form}}{\set{[0,1]\mapsto 0.5}}\\
  &\quad \To
x:\Tref{i}{\Tint{\nu}{\form}}{\set{[0,1]\mapsto 0.5}}
\end{align*}
hold for \(\TE=x\COL 
\Tref{i}{\Tint{\nu}{\form}}{\set{[0,1]\mapsto 1}}\),
\(\form_1 \defeq i=0\Imp \nu=0\), \(\form_2 \defeq i=1\Imp \nu=1\), and \(\form \defeq \form_1\land\form_2\).

In the rule \textsc{T-MkArray} for memory allocation, we assign the full ownership \(1\) for
the relative addresses \(i=0,\ldots,n-1\), and the null ownership \(0\) to other addresses, to prevent
out-of-boundary memory access. In addition, we require that the types of values stored in
the allocated memory region satisfy the \(\Empty\) predicate, as those values are not initialized.
In the rule \textsc{T-Deref} for pointer deference,
we require that the ownership \(r(0)\) of the pointer \(x\) with index \(0\) is positive.
The type \(\ty_y\) of the value stored at \(y\) (with index \(0\)) is
split into \(\ty'\) and \(\ty_x\), and the type of \(x\) in \(e\) is set to \(\ty_x\).
The \emph{strengthening operation} \(\ty \land_x \form\) and
typed equality proposition \(x =_{\ty} y\) are defined by:
\iftwocol
\begin{align*}
  \Tint{\nu}{\form}\land_x\form'&=\Tint{\nu}{\form\land[\nu/x]\form'}\\
  \Tref{i}{\ty}{r}\land_x\form'&=\Tref{i}{\ty}{r}\\ 
  \quad (x=_{\Tint{\nu}{\form}}y)&=(x=y)\\
   (x=_{\Tref{i}{\ty}{r}}y)&=\top
\end{align*}
\else
\begin{alignat*}{2}
  \Tint{\nu}{\form}\land_x\form'&=\Tint{\nu}{\form\land[\nu/x]\form'} &\quad (x=_{\Tint{\nu}{\form}}y)&=(x=y)\\
  \Tref{i}{\ty}{r}\land_x\form'&=\Tref{i}{\ty}{r} &\quad (x=_{\Tref{i}{\ty}{r}}y)&=\top
\end{alignat*}
\fi
The operations \(\ty \land_x \form\) and \(x =_{\ty} y\) are always used together to propagate the information about equalities over integer variables.

In the rule \textsc{T-Assign}, the premise \(r(0)=1\)
ensures that the ownership of the pointer \(x\) (at index \(0\)) is \(1\).
The type of \(y\) is split into \(\ty_1\) for the type of \(y\) after the assignment,
and \(\ty_2\) for the type of the value stored at \(x\) with index \(0\).
In the rule \rn{T-AddPtr} for pointer arithmetic,
the type and ownership of \(y\) is split into those for \(y\) and \(x\), and the indices for
the type and ownership of \(x\) are shifted accordingly.

The rules \rn{T-AliasAddPtr} and \rn{T-AliasDeref} allow us to redistribute
the ownerships of pointers based on alias information.
These rules are analogous to the corresponding rules of \consort{}, but
more complex due to the presence of pointer arithmetic.

The rule \textsc{T-Assert} requires that the asserted formula \(\form\) must be valid
under the current context, where \(\formTE{\TE}\) is defined by:
\begin{alignat*}{2}
  \formTE{\bullet}&=\top &\quad \formTE{\Tint{\nu}{\form}}_y&=[y/\nu]\form\\ 
  \formTE{\TE, x:\ty}&=\formTE{\TE}\land\formTE{\ty}_x &\quad \formTE{\Tref{i}{\ty}{r}}_y&=\top
\end{alignat*}

In the rule \rn{T-Sub}, the premise \(\TE\le \TE'\) means that \(\TE\) is a stronger type assumption
than \(\TE'\) (for example, \(x\COL \Tint{\nu}{\nu=1} \le x\COL \Tint{\nu}{\nu>0}\)).
The rule allows us to strengthen the type environment before an execution of the expression,
and weakens the type and type environment after the execution.

Figure~\ref{fig:type3} gives typing rules for function definitions and programs.
The rule \rn{T-FunDef} checks that the function definition is consistent with
the assumption \(\FE\) on the types of functions.
The conclusion of the rule \rn{T-Funs} means that the function definitions \(D\)
provides a function environment as described by \(\FE\). The rule \rn{T-Prog} is
for the whole program, which checks that \(e\) is well-typed under the function
type environment \(\FE\) provided by \(D\).

\begin{example}
  Recall Example~\ref{ex:init}.
  The else-part of \(e_1\) is typed as shown in Figure~\ref{fig:ex:init}.
  \begin{figure*}
\begin{align*}
  & &&\mbox{// } p:\Tref{i}{\Tint{\nu}{\TRUE}}{\set{[0,x-1]\mapsto 1}}\\
  &p:=0; &&\mbox{// } p:\Tref{i}{\Tint{\nu}{i=0\Imp\nu=0}}{\set{[0,x-1]\mapsto 1}}\\
  &\letexp{j}{1} &&\mbox{// } j:\Tint{\nu}{\nu=1}\\
  &\letexp{q}{p\pplus j} &&\mbox{// } p:\Tref{i}{\Tint{\nu}{i=0\Imp\nu=0}}{\set{0\mapsto 1}}, %
  q:\Tref{i}{\Tint{\nu}{\top}}{\set{[0,x-2]\mapsto 1}}\\
  &\letexp{y}{x-j} && \mbox{// } y:\Tint{\nu}{\nu=x-1}\\
  &\letexp{z}{\mathit{init}(y,q)} &&\mbox{// }p:\Tref{i}{\Tint{\nu}{i=0\Imp\nu=0}}{\set{0\mapsto 1}}, %
  q:\Tref{i}{\Tint{\nu}{0\le i<x-2 \Imp \nu=0}}{\set{[0,x-2]\mapsto 1}}\\
  & \aliasexp{q}{p\pplus j}&&\mbox{// }p:\Tref{i}{\Tint{\nu}{0\le i\le x-1\Imp\nu=0}}
  {\set{[0,x-1]\mapsto 1}}
\end{align*}
\caption{Typing for Example~\ref{ex:init}}
\label{fig:ex:init}
\end{figure*}
\noindent
Here, we have omitted some of the type bindings (such as \(x\COL\INT\)).
The first line indicates that there is initially no information on the value stored at \(p\).
After the assignment \(p:=0\), the type of the value stored at \(p\) is updated accordingly.
After \(\mathbf{let}\ q=p\pplus j\),  the ownership held by \(p\) is split into
\(\set{0\mapsto 1}\) and \(\set{[1,x-1]\mapsto 1}\), and the latter is passed to \(q\).
By using the rule \rn{T-Call}, and assuming that \(\mathit{init}\) has the type given
in Figure~\ref{fig:example1}, the type of \(q\) is updated. On the last line,
the types of \(p\) and \(q\) are merged into the type of \(p\), by using \rn{T-AliasAddPtr}
(and the rule \rn{T-Sub} for subsumption).
\qed
\label{ex:typeinit}
\end{example}

\begin{example}
As mentioned, our type system supports arrays of pointers, which are not supported by \consort{}.
An example of a program manipulating such an array is given in Figure~\ref{fig:exampleMat}.
The function \( \mathit{initMatrix} \) takes integers \( x \), \( y \) and a \( x \times y \) matrix \( p \) (represented as a pointer of a pointer) and initializes it as a zero matrix by using the running example \( \mathit{init} \).
Some important type information is given as comments to help readers understand how \( \mathit{initMatrix} \) is typed.
\qed
\end{example}

\begin{figure*}
  \begin{align*}
    &\mbox{// } \mathit{initMatrix}\COL
      \Tfunarg{x\COL\Tint{\nu}{\top}, y\COL\Tint{\nu}{\top}, p\COL \Tref{i}{\Tref{j}{\Tint{\nu}{\top}}{\set{[0,y-1]\mapsto 1}}}{\set{[0,x-1]\mapsto 1}}}\\
    &\mbox{// }\qquad\qquad\to
      \Tfunret{x\COL\Tint{\nu}{\top}, y\COL\Tint{\nu}{\top}, p\COL \Tref{i}{\Tref{j}{\Tint{\nu}{\form_p}}{\set{[0,y-1]\mapsto 1}}}{\set{[0,x-1]\mapsto 1}}}{\Tint{\nu}{\top}}\\
    &\mbox{// }\qquad\qquad\qquad\qquad \form_p=0\leq i \leq x-1 \land 0\leq j \leq y-1 \Imp \nu=0 \\
    &\mathit{initMatrix}(x, y, p) \{\\
    &\quad
    \begin{aligned}
      &\ifexp{x}{1} && \\
      &\letexp{q}{*p} && \mbox{// } q: \Tref{j}{\Tint{\nu}{\TRUE}}{\set{[0,y-1]\mapsto 1}} \\
      &\letexp{z}{\INIT(y,q)} && \mbox{// } q: \Tref{j}{\Tint{\nu}{0\leq j \leq y-1\Imp \nu=0}}{\set{[0,y-1]\mapsto 1}} \\
      &\aliasexp{q}{*p} && \mbox{// } p : \Tref{i}{\Tref{j}{\Tint{\nu}{i=0 \land 0\leq j \leq y-1 \Imp \nu=0}}{\set{[0,y-1]\mapsto 1}}}{\set{[0,x-1]\mapsto 1}} \\
      &\letexp{i}{1} \letexp{p'}{p\pplus i} &&  \mbox{// } p' : \Tref{i}{\Tref{j}{\Tint{\nu}{\TRUE}}{\set{[0,y-1]\mapsto 1}}}{\set{[0,x'-2]\mapsto 1}} \\
      &\letexp{x'}{x-i} && \\
      &\letexp{z'}{\mathit{initMatrix}(x', y, p')} && \mbox{// } p' : \Tref{i}{\Tref{j}{\Tint{\nu}{i=0 \land 0\leq j \leq y-1 \Imp \nu=0}}{\set{[0,y-1]\mapsto 1}}}{\set{[0,x'-1]\mapsto 1}} \\
      &\aliasexp{p'}{p\pplus i} && \mbox{// } p: \Tref{i}{\Tref{j}{\Tint{\nu}{\form_p}}{\set{[0,y-1]\mapsto 1}}}{\set{[0,x-1]\mapsto 1}} \\
      &0 \} &&\\
    \end{aligned}
  \end{align*}
  \caption{An Example with Arrays of Pointers}
  \label{fig:exampleMat}
\end{figure*}

\input{type-soundness}

%% file: fig-type.tex
\begin{figure*}

    \begin{multicols}{2}
    
    \infrule[T-Int]
        {}{\FE\mid\TE\p n:\Tint{\nu}{\nu=n}\To\TE}

    \infrule[T-Let]
        {\\
         \FE\mid\TE\p e_1:\ty_1\To\TE_1\ \, x\notin\dom(\TE')\\
         \FE\mid\TE_1,x:\ty_1\p e_2:\ty\To\TE'}
        {\FE\mid\TE\p\letexp{x}{e_1}e_2:\ty\To\TE'}
    
    \infrule[T-Var]
        {}{\FE\mid\TE[x:\ty_1+\ty_2]\p x:\ty_1\To\TE[x\update\ty_2]}

    \infrule[T-MkArray]
        {r(i)=\begin{cases}
                1&0\leqq i\leqq n-1\\
                0&otherwise
              \end{cases}\\
         x\notin\dom(\TE')\andalso
         Empty(\ty')\\
         \FE\mid\TE,x: \Tref{i}{\ty'}{r}\p e:\ty\To\TE'}
        {\FE\mid\TE\p\mkarrexp{x}{n}e:\ty\To\TE'}
    
    \end{multicols}
    
    \infrule[T-If]
        {\FE\mid\TE[x\update\Tint{\nu}{\form\land\nu\le 0}]\p e_1:\ty\To\TE'\\
         \FE\mid\TE[x\update\Tint{\nu}{\form\land\nu>0}]\p e_2:\ty\To\TE'}
        {\FE\mid\TE[x:\Tint{\nu}{\form}]\p\ifexp{x}{e_1}e_2:\ty\To\TE'}

    \infrule[T-Call]
        {\\
         \FE(f)=\langle x_1:\ty_1, \dots ,x_n:\ty_n\rangle\rightarrow\langle x_1:\ty_1',\dots ,x_n:\ty_n'\mid \ty\rangle \andalso
         \sigma=[y_1/x_1]\cdots[y_n/x_n]\\
         \FE\mid\TE[y_i\update\sigma\ty_i'],x:\sigma\ty\p e:\ty'\To\TE'\andalso x\notin\dom(\TE') }
        {\FE\mid\TE[y_i:\sigma\ty_i]\p\letexp{x}{f(y_1,\dots,y_n)}e:\ty'\To\TE'}

    \infrule[T-Deref]
        {i:\Tint{\nu}{\nu=0}\p \ty'+\ty_x\approx\ty_y, \ty_y'\approx\ty'\land_y(y=_{\ty'}x)\\
         i:\Tint{\nu}{\nu\ne 0}\p \ty_y'\approx\ty_y\\
         \FE\mid\TE[y\update\Tref{i}{\ty_y'}{r}],x:\ty_x\p e:\ty\To\TE'\andalso x\notin\dom(\TE')\andalso r(0)>0}
        {\FE\mid\TE[y:\Tref{i}{\ty_y}{r}]\p\letexp{x}{*y}e:\ty\To\TE'}
    
    \infrule[T-Assign]
        {i:\Tint{\nu}{\nu=0}\p \ty_x'\approx\ty_2\land_x(x=_{\ty_2}y)\andalso
         i:\Tint{\nu}{\nu\ne 0}\p \ty_x'\approx\ty_x\\
         \FE\mid\TE[x\update\Tref{i}{\ty_x'}{r}][y\update\ty_1]\p e:\ty\To\TE'\andalso r(0)=1}
        {\FE\mid\TE[x:\Tref{i}{\ty_x}{r}][y:\ty_1+\ty_2]\p x:=y;e:\ty\To\TE'}
        
    \infrule[T-AddPtr]
        {x\notin\dom(\TE')\andalso \forall i\in\Z.r_x(i-z)+r_y'(i)=r_y(i)\\
         \FE\mid\TE[y\update\Tref{i}{\ty_1}{r_y'}],x:\Tref{i}{[(i+z)/i]\ty_2}{r_x},z:\Tint{\nu}{\form}\p e:\ty\To\TE'}
        {\FE\mid\TE[y:\Tref{i}{(\ty_1+\ty_2)}{r_y}][z:\Tint{\nu}{\form}]\p\letexp{x}{y\pplus z}e:\ty\To\TE'}

    \infrule[T-AliasDeref]
        {i_2:\Tint{\nu}{\nu=0} \p (\Tref{i_1}{\ty_x}{r_x}+\Tref{i_1}{\ty_y}{r_y})\approx(\Tref{i_1}{\ty_x'}{r_x'}+\Tref{i_1}{\ty_y'}{r_y'})\\
         i_2:\Tint{\nu}{\nu\ne 0}\p \Tref{i_1}{\ty_y}{r_y}\approx\Tref{i_1}{\ty_y'}{r_y'}\\    
         \FE\mid\TE[x\update\Tref{i_1}{\ty_x'}{r_x'}][y\update\Tref{i_2}{\Tref{i_1}{\ty_y'}{r_y'}}{r}]\p e:\ty\To\TE'}
        {\FE\mid\TE[x:\Tref{i_1}{\ty_x}{r_x}][y:\Tref{i_2}{\Tref{i_1}{\ty_y}{r_y}}{r}]\p\aliasexp{x}{*y}e:\ty\To\TE'}
    
    \infrule[T-AliasAddPtr]
        {\forall i\in\Z.r_{x_z}(i)=r_x(i-z), r_{x_z}'(i)=r_x'(i-z)\\
         (\Tref{i}{[(i-z)/i]\ty_x}{r_{x_z}}+\Tref{i}{\ty_y}{r_y})\approx(\Tref{i}{[(i-z)/i]\ty_x'}{r_{x_z}'}+\Tref{i}{\ty_y'}{r_y'})\\
         \FE\mid\TE[x\update\Tref{i}{\ty_x'}{r_x'}][y\update\Tref{i}{\ty_y'}{r_y'}],z:\Tint{\nu}{\form}\p e:\ty\To\TE'}
        {\FE\mid\TE[x:\Tref{i}{\ty_x}{r_x}][y:\Tref{i}{\ty_y}{r_y}][z:\Tint{\nu}{\form}]\p\aliasexp{x}{y\pplus z}e:\ty\To\TE'}
        
    \begin{multicols}{2}
    
    \infrule[T-Assert]
        {\\
         \models\formTE{\TE}\Imp\form\andalso \TE\p_\WF\form\\
         \FE\mid\TE\p e:\ty\To\TE'}
        {\FE\mid\TE\p\assertexp{\form};e:\ty\To\TE'}
    
    \infrule[T-Sub]
        {\\
         \TE\leq\TE'\andalso \TE'',\ty\leq\TE''',\ty'\\
         \FE\mid\TE'\p e:\ty\To\TE''}
        {\FE\mid\TE\p e:\ty'\To\TE'''}
    
    \end{multicols}
    
    \caption{Expression Typing Rules}
        \label{fig:type1}
    
    \end{figure*}
    
    \begin{figure*}
    
    \begin{multicols}{2}
        
    \infrule[S-Int]
    {\models\formTE{\TE}\Imp(\form_1\Imp\form_2)}
    {\TE\p\Tint{\nu}{\form_1}\leq\Tint{\nu}{\form_2}}
    
    \infrule[S-Ref]
    {r_1\geq r_2\andalso \TE\p\ty_1\leq\ty_2}
    {\TE\p\Tref{i}{\ty_1}{r_1}\leq\Tref{i}{\ty_2}{r_2}}
    
    \infrule[S-TyEnv]
    {\forall x\in\dom(\TE').\TE\p\TE(x)\leq\TE'(x)}
    {\TE\leq\TE'}
    
    \infrule[S-Res]
    {\TE,x:\ty\leq\TE',x:\ty'\andalso x\notin\dom(\TE)}
    {\TE,\ty\leq\TE',\ty'}
    
    \end{multicols}   
    
    \centering{\(\TE\p\ty_1\approx\ty_2\) iff \(\TE\p\ty_1\leq\ty_2\) and \(\TE\p\ty_2\leq\ty_1\)}
    
    \caption{Subtyping Rules}
        \label{fig:type2}
        
    \end{figure*}
    
    \begin{figure*}
    
    \infrule[T-FunDef]
        {\FE(f)=\langle x_1:\ty_1, \dots ,x_n:\ty_n\rangle\rightarrow\langle x_1:\ty_1',\dots ,x_n:\ty_n'\mid \ty\rangle \\
         \FE\mid x_1:\ty_1,\dots ,x_n:\ty_n\p e:\ty\To x_1:\ty_1',\dots ,x_n:\ty_n'}
        {\FE\p f\mapsto(x_1,\dots ,x_n)e}
    
    \begin{multicols}{2}
        
    \infrule[T-Funs]
        {\\
         \dom(D)=\dom(\FE)\\
         \forall f\mapsto(x_1,\dots ,x_n)e\in D.\FE\p f\mapsto(x_1,\dots ,x_n)e}
        {\FE\p D}
    
    \infrule[T-Prog]
        {\\
         \FE\p D\andalso \p_\WF\FE\\
         \FE\mid\bullet\p e:\ty\To\TE}
        {\p\langle D,e\rangle}
    
    \end{multicols}
    
    \caption{Program Typing Rules}
        \label{fig:type3}
    
    \end{figure*}

%% file: type-soundness.tex
\subsection{Soundness}
Our type system ensures that any typed program will never experience an assertion failure.
It also guarantees that a well-typed program never gets stuck.
This means that our type system ensures that no out-of-bound access occurs even though memory safety is not the main focus of this paper.
\begin{restatable}[Soundness]{theorem}{soundness}
  \label{th:soundness}
  If \( \vdash \langle D, e \rangle \), then \( \sconfig \emptyset \emptyset e \not\redD^+ \ASSERTFAIL \).
  Moreover, the reduction does not get stuck.
  That is, every well-typed program either (i) halts in a configuration \( \sconfig R H v \), where \( v \) is either a variable or an integer, (ii) runs into an \( \ALIASFAIL \), or (iii) diverges.
\end{restatable}
Note that the type system does not guarantee the absence of \( \ALIASFAIL \).
It is the responsibility of the programmer or the tool used for the alias analysis to ensure that the inserted \( \mathbf{alias} \) expressions are correct.

The proof of soundness is by standard subject reduction and progress lemmas.
We define the notion of \emph{well-typed configurations}, and prove that the reduction preserves the well-typedness.
Roughly speaking, \( \sconfig R H e \) is well-typed if (i) \( \FE \mid \TE \p e : \ty \To \TE' \), (ii) each predicate formula \( \form \) appearing in \( \TE \) satisfies \( \models [R][R(x) / \nu] \form \) (for an appropriate variable \( x \)), and (iii) the sum of the ownerships of pointers pointing to \( \Addr a i \) does not exceed \( 1 \) for every \( \Addr a i \in \dom(H) \).
Details of the proof are found in \Appendix{sec:soundness}.

%% file: infer.tex
\section{Type Inference}
\label{sec:infer}
\newcommand*{\unknown}{\mathbf{??}}
This section describes a type inference procedure, which takes a program with no ownership/refinement type annotations as an input,
and checks whether the program is well-typed in our type system. Thanks to Theorem~\ref{th:soundness} in the previous section,
the procedure serves as a sound (but incomplete) method for automatically checking the lack of assertion failures.
We assume below that the simple type (without ownerships and refinement predicates) of each expression is already known;
it can be automatically obtained by the standard unification-based simple type inference algorithm.

Following \consort{}, the type inference proceeds in two steps: the first phase for ownership inference, and the second phase for refinement
type inference. The first phase of ownership inference requires a significant extension of the ownership inference of \consort{},
as an ownership has been extended from a single constant to a function that may depend on contexts.

\subsection{Ownership Inference}

As mentioned earlier, we restrict an ownership to the form \([l,u]\mapsto o\), which represents
the map \(f\) such that \(f(i)=o\) if \(l\le i\le u\), and \(f(i)=0\) otherwise.
As \(l\) and \(u\) may depend on
variables (recall that in the example in Figure~\ref{fig:example1}, we needed an ownership of the form
\([0,x-1]\mapsto 1\)), for each expression of a reference type, we prepare an ownership template
of the form \([c_0+c_1x_1+\cdots+c_kx_k, d_0+d_1x_1+\cdots+d_kx_k]\mapsto o\)
where \(x_1,\ldots,x_k\) are the integer variables available in the context of the expression,
and \(c_i, d_i, o\) are unknown constants.
The goal of the ownership inference phase is to determine those unknown constants so that
all the ownership constraints in the typing rules in Section~\ref{sec:type} are satisfied.
To this end, we first generate constraints on the unknown constants based on the typing rules,
and then solve them with a help of the Z3 SMT solver.

We illustrate the constraint generation by using the running example.

\begin{example}
  \label{example:infer1}
  Let us consider the part \(p:=0; \letexp{q}{p\pplus 1}\cdots \)  of the example in Figure~\ref{fig:example1}.
We prepare the following template of type environments for each program point:
\begin{align*}
  & \textcolor{blue}{\mbox{//} p\COL \lambda i.\INT\;\REF^{[c_{0,0}+c_{0,1}x,d_{0,0}+d_{0,1}x]\mapsto o_0}}\\
  &p:=0;  \\
  &\textcolor{blue}{\mbox{// } p\COL \lambda i.\INT\;\REF^{[c_{1,0}+c_{1,1}x,d_{1,0}+d_{1,1}x]\mapsto o_1}}\\
  &\letexp{q}{p\pplus 1}  \\
  & \textcolor{blue}{\mbox{// }p\COL \lambda i.\INT\;\REF^{[c_{2,0}+c_{2,1}x,d_{2,0}+d_{2,1}x]\mapsto o_2},}\\& \textcolor{blue}{\mbox{//}
  q\COL \lambda i.\INT\;\REF^{[c_{3,0}+c_{3,1}x,d_{3,0}+d_{3,1}x]\mapsto o_3} }
\end{align*}
\noindent
From \(p:=0\), we generate the following constraints based on the rule \rn{T-Assign}:
\begin{align*}
  &  o_0 = 1 \land \forall x>0. c_{0,0}+c_{0,1}x\le 0 \le d_{0,0}+d_{0,1}x.\\
  &  \forall x>0.c_{0,0}+c_{0,1}x \le c_{1,0}+c_{1,1}x \\&\qquad
  \land d_{1,0}+d_{1,1}x \le d_{0,0}+d_{0,1}x \land o_0\ge o_1.
\end{align*}
Here, the constraint on the first line requires that \(p\) holds the full ownership at index \(0\).
The condition \(x>0\) comes from the fact that the expression occurs in the else clause of the conditional expression
``\(\mathbf{ifnp}\ x\)''.\footnote{Currently, we locally collect such conditions from conditional expressions.
Note that without the condition \(x>0\), the co-efficients \(c_{0,1}\) and \(d_{0,1}\) have to be \(0\). For more precise
analysis, we need to combine the ownership inference and refinement type inference phases. }
The constraint on the second line requires that the ownership after the assignment \(p:=0\) is no greater than that before the assignment.

From \(\letexp{q}{p\pplus 1}{\cdots}\), we generate the following constraints based on the rule \rn{T-AddPtr}:
\begin{align*}
  & \forall x>0. c_{1,0}+c_{1,1}x \le c_{2,0}+c_{2,1}x \\&\qquad
  \land d_{2,0}+d_{2,1}x \le d_{1,0}+d_{1,1}x  \land o_1\ge o_2.\\
  & \forall x>0. c_{1,0}+c_{1,1}x - 1 \le c_{3,0}+c_{3,1}x \\&\qquad
  \land d_{3,0}+d_{3,1}x \le d_{1,0}+d_{1,1}x - 1 \land o_1\ge o_3.\\
  & \forall x>0. o_1\ge o_2+o_3 \lor\\&
  \qquad\qquad d_{3,0}+d_{3,1}x < c_{2,0}+c_{2,1}x - 1 \lor\\&
  \qquad\qquad d_{2,0}+d_{2,1}x - 1 < c_{3,0}+c_{3,1}x. 
\end{align*}
The constraints ensure that the sum of the ownerships of \(p\) and \(q\) after the let-expression is no greater than
the ownership of \(p\) before the expression.
\qed
\end{example}

As in the above example, we prepare a template of a type environment (without refinement type information)
for each program point, and generate the constraints based on the typing rules. 

It remains to solve the constraints on \(c_{i,j},d_{i,j}, o_i\). Unfortunately, the constraints involve nested quantifiers
of the form \(\exists c_{i,j},d_{i,j}, o. \forall x.\varphi\), which cannot be solved efficiently by Z3.
We have thus employed the following heuristic method to reduce the constraints further to the satisfiability
problem for quantifier-free formulas.
Given a formula of the form \(\exists \seq{c}.\forall\seq{x}.\psi(\seq{c},\seq{x})\) (where \(\seq{c}\) and \(\seq{x}\) denote
sequences of variables),
we first under-approximate it to 
\(\exists \seq{c}.\psi(\seq{c},\seq{n}_1)\land \cdots \land \psi(\seq{c},\seq{n}_k)\),
where \(\seq{n}_1,\ldots,\seq{n}_k\) are randomly-chosen values for \(\seq{x}\). We then invoke an SMT solver
(Z3 in our implementation) to obtain a solution \(\seq{c}=\seq{m}\) for the under-approximation.
We then check whether the original formula \(\forall\seq{x}.\psi(\seq{m},\seq{x})\) is valid by calling an SMT solver.
If so, we are done. Otherwise, we increase \(k\) to refine the under-approximation, and repeat until a solution
is found or an under-approximation is found to be unsatisfiable.
It is left for future work to replace this heuristic method with a more systematic, efficient  method.

\subsection{Refinement Inference}

In the phase for refinement type inference, we prepare a predicate variable for each occurrence of integer type,
and generate constraints on the predicate variables. The generated constraints are in the form of
Constrained Horn Clauses (CHC)~\cite{Bjorner15}, which can be  solved by using an off-the-shelf CHC
solver~\cite{DBLP:journals/fmsd/KomuravelliGC16,Eldarica,DBLP:journals/jar/ChampionCKS20}
(though the CHC solving problem is undecidable in general).
This phase is similar to that of \consort{}, except that we have more complex constraints
due to the extension of reference types.

We explain more details by using the running example.

\begin{example}
  Recall the fragment \(p:=0; \letexp{q}{p\pplus 1}\cdots \) of our running example.
  Based on the result of the ownership inference in the previous phase, we can prepare the following template
  for the type environment (which now includes refinement type information) at each program point.
  \begin{align*}
  & \textcolor{blue}{\mbox{//} p\COL \Tref{i}{\Tint{\nu}{P_0(x,i,\nu)}}{[0, x-1]\mapsto 1}}\\
  &p:=0;  \\
    &\textcolor{blue}{\mbox{// } p\COL \Tref{i}{\Tint{\nu}{P_1(x,i,\nu)}}{[0, x-1]\mapsto 1}}\\
  &\letexp{q}{p\pplus 1}  \\
  & \textcolor{blue}{\mbox{// }p\COL \Tref{i}{\Tint{\nu}{P_2(x,i,\nu)}}{[0, 0]\mapsto 1},}\\
  & \textcolor{blue}{\mbox{// }
  q\COL \Tref{i}{\Tint{\nu}{P_3(x,i,\nu)}}{[0, x-2]\mapsto 1}}
\end{align*}
  From the assignment \(p:=0\), we generate the following constraints based on \rn{T-Assign}:
  \begin{align*}
  &  \forall x, i, \nu. x>0\land i=0\land \nu=0 \Imp P_1(x,i,\nu).\\
    &  \forall x, i, \nu. x>0\land 0<i\le x-1\land P_0(x,i,\nu) \Imp P_1(x,i,\nu).
  \end{align*}
  The first and second constraints respectively  reflect the facts that the value of \(*p\) is \(0\) after the assignment,
  and that the values at other addresses are unchanged.
  
  From the let-expression, we obtain the following constraints based on \rn{T-AddPtr}:
  \begin{align*}
    & \forall x,i,\nu.x>0\land P_1(x,0,\nu) \Imp P_2(x,0,\nu)\\
    & \forall x,i,\nu.x>0\land 0\le i\le x-2\land P_1(x,i+1,\nu) \Imp P_3(x,i,\nu).
  \end{align*}
  The second constraint captures the fact that the value at \(q+i\) corresponds to the one at \(p+(i+1)\).

  In addition, we generate the following well-formedness conditions (recall that at an index with no ownership,
  the refinement type \(\ty\) should satisfy \(\Empty(\ty)\)).
  \begin{align*}
    & \forall x,i,\nu. x>0\land (i<0\lor i>x-1)\Imp P_0(x,i,\nu)\\
    & \forall x,i,\nu. x>0\land (i<0\lor i>x-1)\Imp P_1(x,i,\nu)\\
    & \forall x,i,\nu. x>0\land i\ne 0\Imp P_2(x,i,\nu)\\
    & \forall x,i,\nu. x>0\land (i<0\lor i>x-2)\Imp P_3(x,i,\nu)
   \end{align*}
By solving the constraints on predicate variables by using a CHC solver, and substituting a solution for the templates
of types and type environments, we can obtain typing as given in Figure~\ref{fig:example1}.\footnote{We may obtain a slightly different typing, depending on the solution found by a CHC solver; anyway, if the constraints are
satisfiable, we know that the given program is well-typed, hence also that the program does not suffer from any
assertion failures.}
\qed
\end{example}

\subsection{Limitations}

In this subsection, we show programs that cannot be handled by our current method.
The first example given below comes from the restriction on the shape of ownership functions to \([l,u]\mapsto o\).

\begin{example}[Limitation caused by the restriction of ownership functions]
Consider the following program.
\begin{align*}
    &\mkarrexp{p}{3}  &&\mbox{// } p\COL \INT\;\REF^{\set{[0,2]\mapsto 1}}\\
    &\letexp{q_1}{p\pplus 1} &&\mbox{// } q_1\COL \INT\;\REF^{\set{0\mapsto 1}}, p\COL\INT\;\REF^{\set{0\mapsto 1, 2\mapsto 1}}\\
    &p:=1; \\
    &q_1:=2; \\
    &\letexp{q_2}{p\pplus 2} &&\mbox{// } q_2\COL\INT\;\REF^{\set{0\mapsto 1}}, p\COL \INT\;\REF^{\set{0\mapsto 1}}\\
    &q_2:=3; \ ... 
\end{align*}
\noindent
In order to type the program above, as indicated by the comment on the second line, we have to assign to \(p\) the ownership \(r\)
such that \(r(0)=1, r(1)=0, r(2)=1\), which is not supported by our type inference procedure.
\qed
\end{example}

The second example given below identifies the limitation caused by the lack of polymorphic function types.
\begin{example}
  Consider the following variation of the init function, where the first argument \(x\) has been replaced by
  a pointer \(x_p\) to an integer. In this case, the range of the ownership of \(p\) depend on the value pointed to by \(x_p\),
  which cannot be expressed by our syntax of types.
  \begin{align*}
    &\INIT(x_p, p) \{\\
    &\ifexp{{*x_p}}{1} p := 0; \letexp{q}{p\pplus 1} x_p:=*x_p-1; \\
    & \letexp{z}{\INIT(x_p,q)}\qquad\mbox{// }{*x_p}=x-1\\
    &    \aliasexp{q}{p\pplus 1} 1\}
  \end{align*}
  To type the program above, we need to allow polymorphic function types and assign the following type to \(\INIT\): 
\begin{align*}  
&  \forall x.
  \TfunargL{}x_p\COL\Tref{i}{\Tint{\nu}{\nu=x}}{\set{0\mapsto 1}}, \\&\qquad
  p\COL \Tref{i}{\Tint{\nu}{\top}}{\set{[0,x-1]\mapsto 1}}\TfunargR{}\\
  &\to
      \Tfunret{x_p\COL\Tref{i}{\Tint{\nu}{\top}}{\set{0\mapsto 1}}, \\
  &\qquad p\COL \Tref{i}{\Tint{\nu}{0\leq i \leq x-1\Imp\nu=0}}{\set{[0,x-1]\mapsto 1}}\\
  &\qquad}{\Tint{\nu}{\top}}.
\end{align*}
Although it is not difficult to extend the type system itself, the introduction of polymorphic types requires a significant
extension of the type inference procedure, which is left for future work. \qed
\end{example}

%% file: experiment.tex
\section{Experiments}
\label{sec:experiment}
We have implemented a prototype verifier based on the methods outlined in Section~\ref{sec:infer} and conducted preliminary experiments to confirm the effectiveness of our approach.\footnote{The source code of our tool is available at \url{https://github.com/mamizu-git/Extended_ConSORT}}
Our tool takes programs written in the language described in Section~\ref{sec:lang} as inputs.
The tool first inserts some trivial alias annotations automatically, by replacing \(\letexp{x}{y\pplus z}{e}\) with
\(\letexp{x}{y\pplus z}{\letexp{w}{e}{\aliasexp{x}{y\pplus z}{w}}}\).%
\footnote{Note that the addition of alias annotations
only increases the precision of our type-based analysis. As reported below, thanks to the automatic insertion of
alias expressions, no user annotations of alias expressions were required in our experiments.}
The tool then runs the type inference algorithm described in Section~\ref{sec:infer} to statically checks lack of assertion failures.
The current tool does not support nested pointers.
We used Z3~\cite{DBLP:conf/tacas/MouraB08} version 4.11.2 as the backend SMT-solver, and HoIce~\cite{10.1007/978-3-030-02768-1_8} version 1.10.0 as the backend CHC-solver.

The benchmark programs we prepared for evaluation are as follows:
\begin{itemize}
    \item \textbf{Init} : The motivating program shown in Figure~\ref{fig:example1} in Section~\ref{sec:intro}.
    \item \textbf{Sum} : A program to verify that the sum of the elements in an non-negative integer array remains non-negative.
    \item \textbf{Copy-Array} : A program that copies array elements to another array.
    \item \textbf{Add-Array} : A program that takes two integer arrays of the same size and writes the element-wise sum to another array.
\end{itemize}

\noindent
\textbf{Copy-Array} and \textbf{Add-Array} were included as representative programs that deal with multiple arrays.
To make sure that our tool works equally well with different array access patterns, we have prepared a few variants of \textbf{Sum}.
The original \textbf{Sum} (as well as the other programs listed above) accesses each element of the array from the front.
The program \textbf{Sum-Back} reads each element of the array from the back;
\textbf{Sum-Both} accesses each element from both front and back at the same time;
and \textbf{Sum-Div} divides the array into halves and calculates the summation recursively.\footnote{Splitting an array to multiple arrays was not allowed in \consort{}.}
The lengths of arrays are fixed to 1000 in all the programs we described so far.
To see how the length affects the performance of our tool, we also prepared a program \textbf{Init-10} that is essentially the same as \textbf{Init} but uses an array of length 10.
Those benchmark programs are listed in \Appendix{sec:benchmark}.

We compared the performance of our tool with SeaHorn (version 10.0.0-rc0-3bf79a59)~\cite{10.1007/978-3-319-21690-4_20}, an existing CHC-based automatic verifier for C.\footnote{The input programs were manually translated into C for SeaHorn.}
The results of the experiments are shown in Table~\ref{table:experiments}.
All the instances were run on a machine with 16 GB RAM and Apple M2 CPUs, with a time out of 600 seconds.
The table shows the running time in seconds if the program was successfully verified.

From these results, we confirmed that our method can verify non-trivial programs that involve pointer arithmetic within a realistic amount of time.
The results also demonstrate that the ownership functions of the form \([l,u]\mapsto o\) is not too restrictive.

SeaHorn resulted in timeout for all benchmark programs, except for \textbf{Init-10}.
This did not change even if we switched SeaHorn's backend CHC-solver from Spacer~\cite{10.1007/978-3-319-08867-9_2} (which is the default) to HoIce, which our tool uses.
This is probably because SeaHorn models the heap memory as an integer array
and thus the resulting CHC problems involve predicates on integer arrays.
Thanks to our ownership types, the CHC problems generated by our method
are much easier to solve,  involving only integer predicates.

\begin{table}[t]
  \caption{The results of the experiments. All times are in seconds.}
  \label{table:experiments}
  \centering
  \begin{tabular}{lc|cc}\hline
                           & \textbf{Ours}  &  \multicolumn{2}{c}{\textbf{SeaHorn}}        \\
    \textbf{Name}          & \textit{w/HoIce} & \textit{w/HoIce}       & \textit{w/Spacer} \\  
    \hline\hline
    \textbf{Init-10}       & 1.67             & tool error             & 0.37  \\
    \textbf{Init}          & 3.69             & timeout                & timeout \\
    \textbf{Sum}           & 2.01             & tool error             & timeout \\
    \textbf{Sum-Back}      & 0.54             & timeout                & timeout \\
    \textbf{Sum-Both}      & 1.07             & timeout                & timeout \\
    \textbf{Sum-Div}       & 0.83             & tool error             & timeout \\
    \textbf{Copy-Array}    & 31.31            & timeout                & timeout \\
    \textbf{Add-Array}     & 338.30           & timeout                & timeout \\
    \hline
  \end{tabular}
\end{table}

The bottleneck of our approach is the heuristic used for solving the formula of the shape \( \exists \seq{c}.\forall\seq{x}.\psi(\seq{c},\seq{x}) \) in the ownership inference process, which we described in Section~\ref{sec:infer}.
Our tool struggles to solve formulas of this shape and it may time out even if the formula is valid.
This is more likely to occur when the number of pointers handled in a program increases.
Indeed, \textbf{Copy-Array}, which manipulates two pointers, spends most of the time for the iteration that generates the randomly-chosen values \(\seq{n}_1,\ldots,\seq{n}_k\) for \(\seq{x}\).
For the same reason, \textbf{Add-Array}, which manipulates three pointers, needs even longer time to be verified.

As mentioned earlier, it is left for future work to implement more efficient methods for solving formulas with quantifiers of the form \( \exists \forall \).

%% file: rel.tex
\section{Related Work}
\label{sec:related}
We are not aware of so many studies on automated methods for static verification of
low-level programs with pointer arithmetic.

As already mentioned in Section~\ref{sec:intro},
our type system is an extension of that of \consort{}~\cite{DBLP:conf/esop/TomanSSI020},
which did not support  pointer arithmetic. Also, \consort{} supports arrays,
but elements are restricted to integers (thus, arrays of pointers are not supported),
and a pointer to an array must hold the ownership for the entire array
(\cite{DBLP:conf/esop/TomanSSI020}, Section~4.2).

Rondon et al.~\cite{DBLP:conf/popl/RondonKJ10,DBLP:conf/cav/RondonBKJ12} also
developed a refinement type system for a low-level language with refinement types.
It is difficult to compare the approaches directly, as 
their mechanism for strong updates of refinement types is different from
that of \consort{} and ours: instead of using ownership types, their type system
distinguishes between concrete locations and abstract locations,
and allows strong updates for only concrete locations.
To our knowledge, their type system does not support refinement predicates that
depend on an index to a memory region of variable length
(such as ``the \(i\)-th element of a memory region pointed to by \(p\) is no less than \(i\)'').

In the context of separation logic, 
Calcagno et al.~\cite{DBLP:conf/sas/CalcagnoDOY06} proposed a shape analysis method
in the presence of pointer arithmetic. To our knowledge, their method is specialized for
inferring shape information, and cannot be used for verifying general properties of the contents stored in
a memory region.

There are a number of \emph{semi-}automated approaches to verification of
pointer-manipulating programs~\cite{DBLP:journals/pacmpl/PulteMSMSK23,DBLP:conf/pldi/SammlerLKMD021,DBLP:series/natosec/0001SS17,DBLP:journals/pacmpl/0001W023,Wolff}.
While those approaches cover a wider range of pointer operations (such as
XOR~\cite{DBLP:conf/pldi/SammlerLKMD021}) and computation models (including
concurrency),
they require much more human interventions, such as annotations of
loop invariants and pre/post conditions for each recursive function.
Furthermore, the main interests in many of those studies
are in memory safety and pointer races, while our interest is in
the verification of functional correctness
(expressed by using assertions) of  low-level programs.

%% file: concl.tex
\section{Conclusion}
\label{sec:conc}
We have proposed a type system for an imperative language with
pointer arithmetic, as an extension of the ownership refinement
type system of \consort{}~\cite{DBLP:conf/esop/TomanSSI020},
and proved its soundness.
We have implemented a prototype system for fully automated
verification of programs with pointer arithmetic, and confirmed
its effectiveness. Future work includes an extension of
the tool to support the full language (nested pointers, in particular), and an improvement of
the efficiency of the ownership inference procedure.

%% file: wf.tex
\section{Well-formedness Conditions}
\label{sec:wf}
The well-formedness conditions for types and type environments are given in
Figure~\ref{fig:wf}.

\begin{figure}[h]

    \begin{multicols}{2}

    \infrule[Emp-Int]
    {\\}
    {\Empty(\Tint{\nu}{\top})}

    \infrule[Emp-Ref]
    {\Empty(\ty) \andalso \forall i\in\Z.r(i)=0}
    {\Empty(\Tref{i}{\ty}{r})}

    \end{multicols}
    
    \begin{multicols}{3}    

    \infrule[WF-Empty]
    {\\ \\}
    {\pWF \emptyTE}

    \infrule[WF-Env]
    {x\notin dom(\TE)\\ \pWF\TE\andalso \TE\pWF\ty}
    {\pWF \TE, x:\ty}
    
    \infrule[WF-Int]
    {\\\TE\pWF \form}
    {\TE\pWF\Tint{\nu}{\form}}

    \end{multicols}

    \infrule[WF-Ref]
    {\\
     \forall j\in\set{i\mid r(i)=0}\exists \ty'. Empty(\ty'), \emptyTE\p\ty'\approx [j/i]\ty \andalso 
     \forall j\in\Z.\TE\pWF [j/i]\ty\\
     \forall x\in\FV(r).\TE(x) = \Tint{\nu}{\form}, \TE\pWF\form}
    {\TE\pWF \Tref{i}{\ty}{r}}
    
    \infrule[WF-Phi]
    {\\
     \forall x\in \FV(\form)\setminus\set{\nu}.\TE(x)=\Tint{\nu}{\form'}, \TE\pWF\form'}
    {\TE \pWF \form}

    \infrule[WF-Result]
    {\TE\pWF\ty \andalso \pWF\TE}
    {\pWF \ty\To\TE}

    \infrule[WF-FunType]
    {\pWF x_1:\ty_1, \dots ,x_n:\ty_n\\
     \pWF \ty\To x_1:\ty_1', \dots ,x_n:\ty_n'}
    {\pWF \langle x_1:\ty_1, \dots ,x_n:\ty_n\rangle\rightarrow\langle x_1:\ty_1',\dots ,x_n:\ty_n'\mid \ty\rangle}
    
    \infrule[WF-FunEnv]
    {\\\forall f\in dom(\FE).\pWF\FE(f)}
    {\pWF \FE}

    \caption{Well-formedness}
        \label{fig:wf}
    
\end{figure}

%% file: soundness.tex
\section{Proof of Soundness}
\label{sec:soundness}
\newcommand*{\idx}{i}
\newcommand*{\idxTwo}{j}
\newcommand*{\val}{v}
\newcommand*{\addr}{a}
\newcommand*{\nDeref}[2]{\#\mathrm{Deref}(#1, #2)}
\newcommand*{\size}[1]{|#1|}
\newcommand*{\zerofun}{\mathbf{0}}
\newcommand*{\empEnv}{\bullet}
\newcommand*{\Hupd}[3]{#1 \{ #2 \update #3 \}}
\newcommand*{\Rupd}[3]{#1 \{ #2 \update #3 \}}
\newcommand*{\rname}[1]{\textsc{#1}}
\newcommand*{\assnExp}[3]{#1 := #2; #3}
\newcommand*{\Idx}{\mathbb Z}
\newcommand*{\idxAdd}[2]{#1 + #2}
\newcommand*{\fv}[1]{\mathbf{fv}(#1)}
\newcommand*{\hole}{[\,]}
\newcommand*{\subty}{\le}
\newcommand*{\valuation}{\xi}

This section proves the soundness of our type system.
Our proof strategy follows that for the soundness proof of the type system of \consort{}~\cite{DBLP:conf/esop/TomanSSI020}.
At a very high-level, what we prove are (1) a lemma that states (assert) failure cannot be typed, (2) subject reduction (for configurations) and (3) a standard progress lemma:
\begin{restatable}{lemma}{closedWellTypedProgIsWellTypedConf}
  \label{lem:closed-well-typed-prog-is-well-typed-conf}
  If \( \vdash \langle D, e \rangle \), then \( \vdash_D \sconfig \emptyset \emptyset e : \ty \To \TE' \) for some \( \ty \) and \( \TE' \).
\end{restatable}
\begin{restatable}{lemma}{assertFailNeverHappens}
  \label{lem:assert-fail-never-happens}
  If \( \vdash_D \sconfig R H e : \ty \To \TE' \), then \( \sconfig R H e \not \redD \ASSERTFAIL \).
\end{restatable}
\begin{restatable}[Subject reduction]{lemma}{SubjectReduction}
\label{lem:subject-reduction}
  \newcommand*{\confAfterRed}{\mathbf{C}}
  If \( \p_D \sconfig R H e : \tau \To \TE' \) and \( \sconfig R H e \redD \confAfterRed \), then either
  \begin{itemize}
    \item \( \confAfterRed = \ALIASFAIL \) or
    \item \( \confAfterRed = \sconfig {R'} {H'} {e'} \) with \( \p_D \sconfig {R'} {H'} {e'} : \tau \To \TE' \).
  \end{itemize}
\end{restatable}
\begin{restatable}[Progress]{lemma}{Progress}
\label{lem:progress}
Suppose \( \vdash_D \sconfig R H e : \ty \To \TE' \).
Then either
\begin{enumerate}
  \item \( \sconfig R H e \redD \sconfig {R'} {H'} {e'}  \) holds for some configuration \( \sconfig {R'} {H'} {e'}\),
  \item \( \sconfig R H e \redD \ALIASFAIL  \), or
  \item \( \sconfig R H e \) is of the form \( \sconfig R H v \), where \( v \) is some variable \( x \) or an integer \( n \) (and thus cannot be reduced anymore). \label{it:lem:progress:halt}
\end{enumerate}
\end{restatable}
\
The definition of \( \p_D \sconfig R H e : \tau \To \TE' \)  will be given later in this section.

From these lemmas, we obtain the following soundness result by a straightforward inductive argument.
\soundness*
The rest of this section is organized as follows.
Section~\ref{sec:config-typing} defines the judgement \( \p_D \sconfig R H e : \tau \To \TE' \) and proves Lemma~\ref{lem:closed-well-typed-prog-is-well-typed-conf} and~\ref{lem:assert-fail-never-happens}.
Section~\ref{sec:subject-reduction} proves subject reduction after proving some additional lemmas.
Section~\ref{sec:progress} proves progress.

\subsection{Configuration Typing}
\label{sec:config-typing}
Here we define the type judgement for configurations \( \p_D \sconfig R H e : \tau \To \TE' \) and prove some basic properties about well-typed configurations.

\begin{figure}[t]
 \begin{align*}
    \Own{H}{R}{\TE}&\defeq\sum_{x\in dom(\TE)}\own{H}{R(x)}{\TE(x)}\\
    \own{H}{v}{\ty}&\defeq
    \begin{cases}
        \{\Addr \addr  {\idxAdd \idx \idxTwo} \mapsto [R]r(\idxTwo) \mid \idxTwo \in \Idx \} +\sum_{j \in \Idx}\own{H}{H(\AddrArg \addr{\idxAdd \idx \idxTwo})}{[j/i]\ty'} & \text{if \( (\spadesuit) \)} \\
      \zerofun & \text{otherwise}
    \end{cases} \\
    &\qquad  (\spadesuit)\  v = \Addr \addr \idx, \Addr \addr \idx \in dom(H) \text{ and } \ty=\Tref{i}{\ty'}{r}
 \end{align*}
   \begin{align*}
    &\SAT{H}{R}{\TE} \Def \forall x\in dom(\TE). x \in \dom(R) \text{ and } \SATv{H}{R}{R(x)}{\TE(x)}\\
    &\SATv{H}{R}{v}{\ty} \\
    &\quad \Def
    \begin{cases}
      1.\quad v\in\mathbb{Z}, \ty=\Tint{\nu}{\form} \text{ and } \models [R][v/\nu]\form \text{ or}\\
      2. \quad
      \begin{aligned}[t]
        & v = \Addr \addr \idx, \ty=\Tref{i}{\ty'}{r}, \Addr \addr \idx \in dom(H) \text{ and} \\
        &\quad \forall \idxTwo \in\set{\idx \mid [R]r(\idx)>0}. \Addr \addr {\idxAdd \idx \idxTwo} \in \dom(H) \text{ and } \SATv{H}{R}{H(\AddrArg \addr {\idxAdd \idx \idxTwo})} {[\idxTwo / \idx]\ty'}
      \end{aligned}
      &\end{cases} \\
  \end{align*}
  \caption{Definition of the invariants.}
  \label{fig:invariants}
\end{figure}

\begin{definition}[Well-typed configuration]
  Given a configuration \( \sconfig R H e \), we write \( \vdash_D \sconfig R H e : \ty \To \TE' \)  if the following conditions hold for some type environment \( \TE \) and function type environment \( \FE \):
  \begin{enumerate}
    \item \( \FE \mid \TE \vdash e : \ty \to \TE' \) and \( \FE \vdash D \)
    \item for each \( \Addr \addr \idx \in \dom(H) \), \( \Own H R \TE \le 1 \)
    \item \( \SAT H R \TE \)
  \end{enumerate}
  where \( \Own H R \TE \) and \( \SAT H R \TE \) are defined in Figure~\ref{fig:invariants}.
  We say that a configuration \( \sconfig R H e \) is \emph{well-typed} if there exist \( \ty \) and \( \TE' \) such that \( \vdash_D \sconfig R H e : \ty \To \TE' \).
\end{definition}
The expression \( \zerofun \) in the definition of \( \own H \val \ty \) represents the \emph{constant function that always returns \( 0 \)}, and the function \(  \{\Addr \addr  {\idxAdd \idx \idxTwo} \mapsto r(\idxTwo) \mid \idxTwo \in \Idx \} \) returns \( 0 \) if the argument is not of the form \( \Addr \addr \idxTwo \).
The addition of functions is defined pointwise.
As we already mentioned, \( [R] \form \) represents the formula obtained by replacing each variable \( x \) (of type int) in \( \varphi \) with \( R(x) \).
Formally, it is inductively defined as follows:
\begin{align*}
  [\emptyset]\form\defeq \form  \qquad [R\set{y\mapsto n}]\form\defeq [R][n/y]\form \qquad [R\set{y\mapsto \Addr \addr \idx}] \form \defeq [R]\form.
\end{align*}
Similarly, \( [R]r \) is the ownership function obtained by instantiating the free variables of \( r \) using \( R \).
Hence, \( [R] r \) is a function from \( \mathbb Z \) to \( [0, 1]\).
To simplify the, notation we often omit \( [R] \) and write \( r \) even if we mean \( [R] r \).

Now we briefly explain \( \mathbf{Own}\) and  \( \mathbf{SAT} \).
The function \( \own H \val \ty  \) maps an address \( \Addr \addr \idx \) to the  \emph{hereditary ownership} of \( \Addr \addr \idx \) that \( \val : \ty \) owns.
By hereditary, we mean that if \( \val \) is a pointer of a pointer, then \( \val \) also owns the ownership of \( *\val \), and so on.\footnotemark
\footnotetext{Dereferencing \( *\val \) corresponds to \( H(\AddrArg \addr \idx ) \) provided that \( v = \Addr \addr \idx \).}
It should be emphasized that \( \own H {\Addr \addr \idx} {\Tref \idx \ty r} \) is not focusing on a single memory cell whose address is \( \Addr \addr \idx \), but is describing the hereditary ownership that the chunk of memory starting from \( \Addr \addr \idx \) has.
On the other hand, intuitively, \( \SATv H R \val \ty \) checks whether ``\( \val \) has the type \( \ty \)'' under \( H \) and \( R \).
For example, \( \SATv H R \val {\Tint \nu \form} \) holds if and only if \( \val \) is an integer that satisfies the predicate \( [R]\form \).
In case \( \ty = \Tref \idx {\ty'} r \), \( \mathbf{SATv} \) not only checks that \( \val \) is an address, but also checks that out-of-bounds access does not occur.
The conditions \( \forall \Addr \addr \idx \in \dom(H) . \Own H R \TE (\AddrArg \addr \idx) \le 1 \) and \( \SAT H R \TE \) are the two invariants that will not change during the reduction.
(Note, however, that the type environment \( \TE \) appearing in these expressions may change during reductions.)

We now prove Lemma~\ref{lem:closed-well-typed-prog-is-well-typed-conf} and~\ref{lem:assert-fail-never-happens},~i.e all the lemmas we used to prove the soundness, except for subject reduction.
\closedWellTypedProgIsWellTypedConf*
\begin{proof}
  First, observe that \( \Own H R \empEnv = \zerofun \) and \( \SAT H R \empEnv \) holds for any heap \( H \) and register file \( R \); in particular, we have \( \Own \emptyset  \emptyset \empEnv = \zerofun \) and \( \SAT \emptyset \emptyset \empEnv \) holds.
  Since \( \vdash \langle D, e \rangle\), we have \( \FE \mid  \empEnv \p e : \ty \To \TE' \) for some \( \ty \) and \( \TE' \), and thus \( \vdash_D  \sconfig \emptyset \emptyset e : \ty \To \TE' \).
\end{proof}
The claim that \( \ASSERTFAIL \) never happens (Lemma~\ref{lem:assert-fail-never-happens}) should be intuitively evident as \( \SAT H R \TE \) ensures that all the refinement predicates are valid.
To formally prove this claim, we also need to prove that \( \SAT H R \TE \) is preserved by subtyping.\footnote{We also prove that \( \Own H R \TE(\AddrArg \addr \idx) \le 1 \) is preserved by subtyping even if this is not needed to prove Lemma~\ref{lem:assert-fail-never-happens}. This is a property that we use later in this section.}
\begin{lemma}
  \label{lem:SAT-implies-TE-is-valid}
  If \( \SAT H R \TE \), then \( \models [R] \formTE \TE \). %
\end{lemma}
\begin{proof}
  By the definition of \( \formTE \TE \), it suffices to show that \( \TE(x) = \Tint \nu \form \) implies \( \models [R] [x /\nu]\form \).
  This clearly is the case because \( \SATv H R {R(x)} {\Tint \nu \form} \) implies \( \models [R][R(x) / \nu ] \form (=[R][x/\nu] \form) \).
\end{proof}

\begin{lemma}
  \label{lem:subty-preserves-SATv}
  If \( \TE \vdash \ty \subty \ty' \), \( \models [R]\formTE \TE \) and \( \SATv H R \val \ty \) then \( \SATv H R \val {\ty'} \).
\end{lemma}
\begin{proof}
  By induction on the derivation of \( \TE \vdash \ty \subty \ty' \).
  The base case is when \( \ty = \Tint \nu \form \) and \( \ty' = \Tint \nu {\form'} \) with \( \models \formTE \TE \Imp \form \Imp \form' \).
  This means that we have \( \models [R] \formTE \TE \Imp  [R][\val/\nu] \form \Imp [R][\val/\nu] \form'\), by taking \( R\{ \nu \mapsto v \} \) as the valuation.
  By \( \SATv H R \val \ty \), we have \(  [R][\val/\nu] \form \).
  Together with \( \models [R] \formTE \TE \), we have \( [R][\val / \nu] \form'\), i.e.~\( \SATv H R \val {\ty'} \).

  The inductive case is the case where \( \ty = \Tref \idx {\ty_1} {r_1} \) and \( \ty' = \Tref \idx {\ty_2} {r_2} \) with \( r_1 \ge r_2 \) and \( \TE \p \ty_1 \subty \ty_2 \).
  Assume that \( \val \) is an address \( \Addr \addr \idx \); otherwise it is trivial.
  Our goal is to show that \( \SATv H R {H(\addr, \idxAdd \idx \idxTwo)} {[\idxTwo/\idx]\ty_2} \) for \( \idxTwo \) such that \( r_2(\idxTwo) > 0 \).
  Note that if \( r_2(\idxTwo) > 0 \) then so is \( r_1(\idxTwo) \) because \( r_1 \ge r_2 \).
  From \( \SATv H R \val \ty\), we have \( \SATv H R {H(\addr, \idxAdd \idx \idxTwo)} {[\idxTwo/\idx]\ty_1} \).
  By applying the induction hypothesis we obtain \( \SATv H R {H(\addr, \idxAdd \idx \idxTwo)} {[\idxTwo/\idx]\ty_2} \).\footnote{Strictly speaking, we also need a substitution lemma for the subtyping relation to show \( \TE \vdash [\idxTwo / \idx] \ty_1 \subty [\idxTwo / \idx]\ty_2\). This can be proved by induction on the subtyping relation.}
\end{proof}

\begin{lemma}
\label{lem:subty-preserves-own}
   If \( \TE \vdash \ty \subty \ty' \) and \( \models [R]\formTE \TE \), then \( \own H v \ty \ge \own H v {\ty'} \).
\end{lemma}
\begin{proof}
  By induction on the derivation of \( \TE \vdash \ty \subty \ty' \).
  The base case is trivial as the two ownership functions are the constant function \( \zerofun \).
  The inductive case is the case where \( \ty = \Tref \idx {\ty_1} {r_1} \) and \( \ty' = \Tref \idx {\ty_2} {r_2} \) with \( r_1 \ge r_2 \) and \( \TE \p \ty_1 \subty \ty_2 \).
  We suppose that \( \val \) is an address \( \Addr \addr \idx \) because, if not, it is trivial.
  We have
  \begin{align*}
    \own H {\Addr \addr \idx} \ty
    &= \{\Addr \addr  {\idxAdd \idx \idxTwo} \mapsto r_1(\idxTwo) \mid \idxTwo \in \Idx \} +\sum_{j \in \Idx}\own{H}{H(\AddrArg \addr{\idxAdd \idx \idxTwo})}{[j/i]\ty_1} \tag{by def.} \\
    &\ge \{\Addr \addr  {\idxAdd \idx \idxTwo} \mapsto r_1(\idxTwo) \mid \idxTwo \in \Idx \} +\sum_{j \in \Idx}\own{H}{H(\AddrArg \addr{\idxAdd \idx \idxTwo})}{[j/i]\ty_2} \tag{by I.H.} \\
    &\ge \{\Addr \addr  {\idxAdd \idx \idxTwo} \mapsto r_2(\idxTwo) \mid \idxTwo \in \Idx \} +\sum_{j \in \Idx}\own{H}{H(\AddrArg \addr{\idxAdd \idx \idxTwo})}{[j/i]\ty_2} \tag{since \(r_1 \ge r_2\)} \\
    &= \own H {\Addr \addr \idx} {\ty'}
  \end{align*}
as desired.
\end{proof}

\begin{lemma}
  \label{lem:subty-preserves-Own-and-SAT}
  Suppose that \( \TE \subty \TE' \) and \( \SAT H R \TE \).
  \begin{enumerate}
    \item We have \( \SAT H R {\TE'} \) and
    \item if \( \Own H R \TE(\AddrArg \addr \idx) \le 1  \),  then \( \Own H R {\TE'}(\AddrArg \addr \idx) \le 1 \) for all \( \Addr \addr \idx \in \dom(H) \).
  \end{enumerate}
\end{lemma}
\begin{proof}
  \noindent
  \begin{enumerate}
    \item Follows from Lemma~\ref{lem:SAT-implies-TE-is-valid} and~\ref{lem:subty-preserves-SATv}.
    \item Follows from Lemma~\ref{lem:SAT-implies-TE-is-valid} and~\ref{lem:subty-preserves-own}.
  \end{enumerate}
\end{proof}
\assertFailNeverHappens*
\begin{proof}
  Suppose that \( \sconfig R H  {\assertexp{\form};e_0} \redD \ASSERTFAIL \).
  Then we must have \(\not\models [R]\form \).
  Our goal is to show that \( \models [R]\form \), which leads to a contradiction.
  Since \( \p_D \sconfig R H {\assertexp \form; e_0}\), we must have \( \FE \mid \TE \p {\assertexp \form; e_0} : \ty \To \TE' \) and \( \SAT H R \TE \) for some \( \TE \).
  By inversion on the typing of \( \assertexp{\form};e_0\), we have \( \models \formTE {\TE_0} \To \form \) for some \( \TE_0 \) such that \( \TE \subty \TE_0 \).
  Using Lemma~\ref{lem:subty-preserves-Own-and-SAT} with \( \SAT H R \TE \), we obtain \( \SAT H R {\TE_0}\).
  Therefore, we have \( \models [R] \formTE {\TE_0} \) by Lemma~\ref{lem:SAT-implies-TE-is-valid}, and thus, \( \models [R] \form \).
\end{proof}

\subsubsection*{Auxiliary lemmas about the invariants}
The rest of this subsection is devoted to the proofs of auxiliary lemmas about the invariants that are used to prove subject reduction.
All the lemmas proved here are quite simple.
We do not prove any lemmas that have to do with the dynamics of a program; such lemmas will be proved later in Section~\ref{sec:subject-reduction}.
\begin{lemma}
  \label{lem:own-zero-impl-no-overlap}
  Suppose that \( \own H {\Addr {\addr'} {\idx'}} {\Tref \idx \ty r}(\AddrArg \addr \idx) = 0 \) and \( r(\idxTwo) > 0 \).
  Then \( \Addr \addr \idx \neq \Addr {\addr'} {\idxAdd {\idx'} \idxTwo} \).
\end{lemma}
\begin{proof}
  Obvious from the definition of \( \mathbf{own}\).
\end{proof}

By abuse of notation, we write \( \top \) for types \( \ty \) that satisfy \( \Empty(\ty) \).
Moreover, we define \(\top_0 \defeq  \Tint \nu \top \), \(\top_m \defeq \Tref{i}{\top_{m-1}}{\zerofun} \).

\begin{lemma}
  \label{lem:own-top-is-zero}
  For all \( H \), \( \val \) and \( n \) we have \( \own H \val {\top_n} = \zerofun \).
\end{lemma}
\begin{proof}
  By induction on \( n \).
\end{proof}

\begin{lemma}
  \label{lem:SATv-top}
  For any \( H \), \( R \) and \( \val \), we have that
  \begin{enumerate}
    \item if \( \val \) is an integer, then \( \SATv H R \val {\top_0} \), and
    \item if \( \val \) is an address with \( \val \in \dom(H)\), then \( \SATv H R \val {\top_n} \) for any \( n \ge 1 \).
  \end{enumerate}
\end{lemma}
\begin{proof}
  By induction on \( n \).
\end{proof}

Lemma~\ref{lem:ty-equiv-preserves-own-and-SATv} shows that type equivalence preserves \( \mathbf{own} \) and \( \mathbf{SATv} \) and Lemma~\ref{lem:strengthning-preserves-own} states that strengthening operation does not affect \( \mathbf{own} \).

\begin{lemma}
  \label{lem:ty-equiv-preserves-own-and-SATv}
  Suppose that \( \TE \vdash \ty_1 \approx \ty_2 \) and \( \models [R] \formTE \TE \). Then
  \begin{enumerate}
    \item \( \own H \val {\ty_1} = \own H \val {\ty_2} \) and
    \item \( \SATv H R \val {\ty_1} \) if and only if \( \SATv H R \val {\ty_2} \).
  \end{enumerate}
\end{lemma}
\begin{proof}
  By Lemma~\ref{lem:subty-preserves-SATv} and~\ref{lem:subty-preserves-own}.
\end{proof}

\begin{lemma}
  \label{lem:strengthning-preserves-own}
  For any \( \form \) and \( x \), we have \( \own H \val \ty \) = \( \own H \val {\ty \land_x \form }\).
\end{lemma}
\begin{proof}
  By straightforward induction on the structure of \( \tau \).
\end{proof}

The following lemmas say that \( \mathbf{own} \) and \( \mathbf{SATv} \) can be ``split'' in accordance with type splitting.
\begin{lemma}
  \label{lem:own-ty-add}
  If \(\ty=\ty_1+\ty_2\), we have \(\own H \val \ty = \own H \val {\ty_1} + \own H \val {\ty_2}\).
\end{lemma}
\begin{proof}
  By induction on the rules used to derive  \( \ty= \ty_1 + \ty_2 \).
  The base case where \( \ty_1 = \Tint \nu {\form_1} \) and \( \ty_2 = \Tint \nu {\form_2} \) is trivial because \( \own H \val {\ty_1 + \ty_2} = \own H \val {\ty_1} = \own H \val {\ty_2} = \zerofun \).

  The inductive case is the case where \( \ty_1 = \Tref \idx {\ty'_1} {r_1} \), \( \ty_2 = \Tref \idx {\ty'_2} {r_2} \) and \( \ty = \Tref \idx {\ty'_1 + \ty'_2} {r_1 + r_2} \).
  We only consider the case where \( \val \) is an address \( \Addr \addr \idx \) and is in \( \dom(H)\); remaining cases are trivial as the \( \mathbf{own} \) functions become \( \zerofun \).
  We have
  \begin{align*}
    &\own H \val {\ty_1} + \own H \val {\ty_2}  \\
    &= \{\Addr \addr {\idxAdd \idx \idxTwo} \mapsto r_1(\idxTwo) \mid \idxTwo \in \Idx \}  + \sum_j\own{H}{H(\AddrArg \addr {\idxAdd \idx \idxTwo})}{[\idxTwo/\idx]\ty'_1} \\
    &\phantom{=} + \{\Addr \addr {\idxAdd \idx \idxTwo} \mapsto r_2(\idxTwo) \mid \idxTwo \in \Idx \}  + \sum_j\own{H}{H(\AddrArg \addr {\idxAdd \idx \idxTwo})}{[\idxTwo/\idx]\ty'_2} \\
    &= \{\Addr \addr {\idxAdd \idx \idxTwo} \mapsto (r_1 + r_2)(\idxTwo) \mid \idxTwo \in \Idx \} \\
    &\phantom{=} + \sum_j \left( \own{H}{H(\AddrArg \addr {\idxAdd \idx \idxTwo})}{[\idxTwo / \idx]\ty'_1} + \own{H}{H(\AddrArg \addr {\idxAdd \idx \idxTwo})}{[\idxTwo / \idx]\ty'_2} \right) \\
    &= \{\Addr \addr {\idxAdd \idx \idxTwo} \mapsto (r_1 + r_2)(\idxTwo) \mid \idxTwo \in \Idx \} \\
    &\phantom{=} + \sum_\idxTwo \own{H}{H(\AddrArg \addr {\idxAdd \idx \idxTwo})}{[\idxTwo / \idx](\ty'_1 + \ty'_2)} \tag{by I.H.} \\
    &= \own H \val {\Tref \idx {\ty'_1 + \ty'_2} {r_1 + r_2}} =\own H \val \ty.
  \end{align*}
 \end{proof}

\begin{lemma}
  \label{lem:SATv-ty-add}
  If \(\ty=\ty_1+\ty_2\), we have \(\SATv{H}{R}{v}{\ty}\) iff \(\SATv{H}{R}{v}{\ty_1}\) and \(\SATv{H}{R}{v}{\ty_2}\).
\end{lemma}
\begin{proof}
  By induction on the rules used to derive \( \ty= \ty_1 + \ty_2 \).
  The base case is the case where \( \ty_1 = \Tint \nu {\form_1} \), \( \ty_2 = \Tint \nu {\form_2} \) and \( \ty = \Tint \nu {\form_1 \land \form_2} \).
  From the definition of \(\mathbf{SATv} \), it suffices to show that \( \models [R][v / \nu] (\form_1 \land \form_2)  \) iff \( \models [R][v / \nu] \form_1  \) and \( \models [R][v / \nu] \form_2  \).
  This is obvious because \( [R][v / \nu] (\form_1 \land \form_2)  = [R][v / \nu] \form_1 \land [R][v / \nu] \form_2 \).

  We now consider the inductive case where \( \ty_1 = \Tref \idx {\ty'_1} {r_1} \), \( \ty_2 = \Tref \idx {\ty'_2} {r_2} \) and \( \ty = \Tref \idx {\ty'_1 + \ty'_2} {r_1 + r_2} \).
  We only prove the only if direction; the other direction is somewhat symmetric.
  Since \( \SATv H R v \ty\),  \( \val \) must be an address \( \Addr \addr \idx \) and we have
  \begin{gather}
    \Addr \addr \idx \in \dom(H) \\
    \Addr \addr {\idxAdd \idx \idxTwo} \in \dom(H) \text{ and  } \SATv H  R  {H(\AddrArg \addr {\idxAdd \idx \idxTwo})} {[\idxTwo/\idx] (\ty'_1 + \ty'_2)} \qquad (\text{for \( \idxTwo \) such that \( (r_1 + r_2)(\idxTwo) > 0 \)})
    \label{eq:SATv-add-ref}
  \end{gather}
  From~\eqref{eq:SATv-add-ref} and the induction hypothesis, we have
  \begin{align*}
    \SATv H  R  {H(\AddrArg \addr {\idxAdd \idx \idxTwo})} {[\idxTwo/\idx] \ty'_1 } \quad \text{and}  \quad \SATv H  R  {\Addr \addr {\idxAdd \idx \idxTwo}} {[\idxTwo/\idx] \ty'_2 }
  \end{align*}
    for each \( j \) that satisfies \( (r_1 + r_2)(j) > 0 \).
  In particular, for \( l \in \{1, 2\}\), we have  \( \SATv H  R  {\Addr \addr {\idxAdd\idx \idxTwo}} {[\idxTwo/\idx] \ty'_l } \) and \( \Addr \addr {\idxAdd \idx \idxTwo} \in \dom(H) \) for each \( \idxTwo \) such that \( r_l(\idxTwo) > 0 \).
  Hence, we have \( \SATv H  R  v {\Tref \idx {\ty'_1 } {r_1}} \) and \( \SATv H  R  v {\Tref \idx {\ty'_2 } {r_2}} \) as desired.
\end{proof}

\subsection{Proof of Subject Reduction}
\label{sec:subject-reduction}
This subsection is divided into two parts: the first part proves additional lemmas about \( \mathbf{own} \) and \( \mathbf{SATv}\), and the second part gives the main proof of subject reduction.
As usual, the main part of the proof of subject reduction is an inductive proof on the transition relation.
The lemmas we prepare in the first part are lemmas that are used in the cases for the transition rules that update the heap \( H \) or the register file \( R \) (e.g.~\rname{R-Assign}).
\subsubsection{Preservation properties under modifications of heaps and register files}
We define a relation that expresses that a heap is obtained by modifying a single cell of another heap.
\begin{definition}
  Two heaps \(H\) and \(H'\) are \emph{equivalent modulo \( \Addr \addr \idx \)}, written \(H\approx_{\Addr \addr \idx} H'\) if:
  \begin{enumerate}
    \item \(dom(H)=dom(H')\)
    \item \(\forall \Addr {\addr'} {\idx'} \in dom(H). \Addr {\addr'} {\idx'} \neq \Addr \addr \idx \Imp H(\AddrArg {\addr'} {\idx'})=H'(\AddrArg {\addr'} {\idx'})\)
  \end{enumerate}
\end{definition}

Of course, \( \own H  \val \ty \) nor \( \SATv H R {\val} {\ty} \) is not preserved under modifications of the heap \( H \) if \( \val \) (directly or indirectly) refers to the modified address.
However, if \( \val \) does not own the cell that has been modified, then \( \mathbf{own} \) and \( \mathbf{SATv} \) are preserved.
\begin{lemma}
  \label{lem:heap-upd-own-preservation}
  If \( H \approx_{\Addr \addr \idx} H' \) and \( \own H \val \ty (\AddrArg \addr \idx) = 0 \), then \( \own H \val \tau = \own {H'} \val \tau \).
\end{lemma}
\begin{proof}
  \newcommand{\idxpj}{\idxAdd {\idx'} \idxTwo}
  By induction on the structure of \( \ty \).
  The base case \( \ty = \Tint \nu \form \) trivially holds because, in this case, \( \own H \val \tau = \zerofun = \own {H'} \val \tau \).

  The inductive case is the case where \( \ty = \Tref \idx {\ty'} r \).
  We only consider the case where \( \val \) is an address \( \Addr {\addr'} {\idx'} \) and \( \Addr {\addr'} {\idx'} \) is in \( \dom(H) \) (which is equal to \( \dom(H')\)); the other cases are trivial because the functions become zero functions.
  We have
  \begin{align*}
    &\own H {\Addr {\addr'} {\idx'}} \ty  \\
    &= \{ \Addr {\addr'} \idxpj \mapsto r(\idxTwo) \mid \idxTwo \in \Idx \} +\sum_j\own{H}{H(\AddrArg {\addr'} \idxpj )}{[\idxTwo / \idx]\ty'} \tag{by def.} \\
    &= \{ \Addr {\addr'} \idxpj \mapsto r(\idxTwo) \mid \idxTwo \in \Idx \} +\sum_{j: r(j)> 0}\own{H}{H(\AddrArg {\addr'} \idxpj)}{[\idxTwo / \idx]\ty'} \tag{by well-formedness and Lemma~\ref{lem:own-top-is-zero}} \\
    &= \{\Addr {\addr'} \idxpj \mapsto r(\idxTwo) \mid \idxTwo \in \Idx \} +\sum_{j: r(j)> 0}\own{H'}{H(\AddrArg {\addr'} \idxpj)}{[\idxTwo / \idx]\ty'} \tag{by I.H.} \\
    &= \{\Addr {\addr'} \idxpj \mapsto r(\idxTwo) \mid \idxTwo \in \Idx \} +\sum_{j: r(j)> 0}\own{H}{H'(\AddrArg {\addr'} \idxpj)}{[\idxTwo / \idx]\ty'} \tag{by \( H \approx_{\Addr \addr \idx} H' \) and Lemma~\ref{lem:own-zero-impl-no-overlap}} \\
    &= \own H {\Addr {\addr'} {\idx'}} \ty  \tag{by a reasoning symmetric to the first two equations}
  \end{align*}
  To justify the above equational reasoning, we are left to check that, for any \( \idxTwo \) such that \( r(\idxTwo) > 0\), \( \own{H}{H(\AddrArg {\addr'} \idxpj)}{[\idxTwo / \idx]\ty'} (\AddrArg \addr \idx) = 0 \) so that we can indeed apply the induction hypothesis.
  If, to the contrary, \( 0 < \own{H}{H(\AddrArg {\addr'} \idxpj)}{[\idxTwo / \idx]\ty'} (\AddrArg \addr \idx) \), then we have \( \own H {\Addr {\addr'} {\idx'}} \ty (\AddrArg \addr \idx) > 0 \), which contradicts to \( \own H {\Addr {\addr'} {\idx'}} \ty (\AddrArg \addr \idx) = 0 \).
\end{proof}

\begin{lemma}
  \label{lem:hp-upd-SATv-preservation}
  If \(H \approx_{\Addr \addr \idx} H'\), \(\own{H}{v}{\ty}(\AddrArg \addr \idx) =0\) and \(\SATv{H}{R}{v}{\ty}\), then \(\SATv{H'}{R}{v}{\ty}\).
\end{lemma}
\begin{proof}
  \newcommand*{\typj}{[\idxTwo / {\idx'}] \ty'}
  \newcommand*{\idxpj}{\idxAdd {\idx'} \idxTwo}
  By induction on the structure of \( \ty \).

  The base case where \( \ty = \Tint \nu \form \) is trivial because the definition of \( \SATv H R v \ty \) does not depend on the heap \( H \).

  The step case is the case where \( \ty = \Tref \idx {\ty'} r \) and \( \val \) is an address \( \Addr {\addr'} {\idx'} \).
  Take \( \idxTwo \) such that \( r(\idxTwo) > 0 \); note that \( \Addr {\addr'} \idxpj \neq \Addr \addr \idx \) by Lemma~\ref{lem:own-zero-impl-no-overlap}.
  Our goal is to show that \( \SATv {H'} R {H'(\AddrArg {\addr'} \idxpj)} \typj \) because this would imply \( \SATv {H'} R {\Addr {\addr'} {\idx'}} \ty \).
  Since \( \own H  {\Addr {\addr'} {\idx'}} \ty  (\AddrArg \addr \idx) = 0 \), it must be the case that \( \own H  {H(\AddrArg {\addr'} \idxpj)} \typj(\AddrArg \addr \idx) = 0 \).
  We also have \( \SATv H R {H(\AddrArg {\addr'} \idxpj)} \typj \) from \( \SATv H R {\Addr {\addr'} {\idx'}} \ty \).
  Therefore, we have \( \SATv {H'} R {H(\AddrArg {\addr'} \idxpj)} \typj  \)  by the induction hypothesis.
  We also have \( H(\AddrArg {\addr'} \idxpj) = H'(\AddrArg {\addr'}  \idxpj)\) because \( H \approx_{\Addr \addr \idx} H' \) and \( \Addr {\addr'} \idxpj \neq \Addr \addr  \idx\).
  Hence, \( \SATv {H'} R {H'(\AddrArg {\addr'} \idxpj)} \typj \) holds.
\end{proof}

Now we show that \( \mathbf{own} \) and \( \mathbf{SATv} \) are preserved under extension of heaps.
\begin{lemma}
  \label{lem:heap-extension-preserves-own}
  If \( \SATv H R \val \ty \), then for any \( \val' \) and address \( \Addr {\addr'} {\idx'} \) such that \( \Addr {\addr'} {\idx'} \notin \dom(H) \) we have \( \own H \val \ty = \own {\Hupd H {\Addr {\addr'} {\idx'}} {\val'}} \val \ty \).
\end{lemma}
\begin{proof}
  \newcommand*{\tyref}{\Tref \idx {\ty'} r }
  \newcommand*{\HTwo}{\Hupd H {\Addr {\addr'} {\idx'}}  {\val'}}
  \newcommand*{\raddr}{\Addr \addr \idx}
  \newcommand*{\idxij}{\idxAdd \idx \idxTwo}
  By induction on the structure of \( \ty \).
  The case where \( \ty = \Tint \nu \form \) is trivial because \( \own H \val \ty = \own \HTwo \val \ty  = \zerofun \) by the definition of \( \mathbf{own} \).

  Now we consider the case \( \ty = \tyref \).
  From the assumption \( \SATv H R \val \ty \), it must be the case that \( \val \) is an address \( \raddr \) such \( \raddr \in \dom(H) \) (which implies that \( \raddr \in \dom(\HTwo) \).
  We also have
  \begin{equation}
    \label{eq:lem:heap-extension-preserves-own:assumption-for-ind}
    \SATv H R {H(\AddrArg \addr \idxij)} {[\idxTwo /\idx] \ty'} \qquad (\text{for \( \idxTwo \) such that \( r(\idxTwo) > 0 \)})
  \end{equation}
  by the definition of \( \mathbf{SATv}\).
  We, therefore,  have
  \begin{align*}
    &\own H \raddr \tyref \\
    &= \{\Addr \addr \idxij \mapsto r(\idxTwo) \mid \idxTwo \in \Idx \} + \sum_\idxTwo \own{H}{H(\AddrArg \addr \idxij)}{[\idxTwo / \idx]\ty'} \tag{by def.} \\
    &= \{\Addr \addr \idxij \mapsto r(\idxTwo) \mid \idxTwo \in \Idx \} + \sum_{\idxTwo: r(\idxTwo) > 0} \own{H}{H(\AddrArg \addr \idxij)}{[\idxTwo / \idx]\ty'} \tag{by well-formedness and Lemma~\ref{lem:own-top-is-zero}} \\
    &= \{\Addr \addr \idxij \mapsto r(\idxTwo) \mid \idxTwo \in \Idx \} \\
    &\phantom{=} + \sum_{\idxTwo: r(\idxTwo) > 0} \own \HTwo {\HTwo(\AddrArg \addr \idxij)}{[\idxTwo / \idx]\ty'} \tag{by I.H. with~\eqref{eq:lem:heap-extension-preserves-own:assumption-for-ind}} \\
    &= \own \HTwo \raddr \tyref.
  \end{align*}
\end{proof}

\begin{lemma}
\label{lem:heap-extension-preserves-SATv}
  If \( \SATv H R \val \ty \), then for any \( \val' \) and \( \Addr {\addr'} {\idx'} \) such that \( \Addr {\addr'} {\idx'} \notin \dom(H) \) we have \( \SATv {\Hupd H {\Addr {\addr'} {\idx'}} {\val'}} R \val \ty \).
\end{lemma}
\begin{proof}
  \newcommand*{\HTwo}{\Hupd H {\Addr {\addr'} {\idx'}} {\val'}}
  \newcommand*{\raddr}{\Addr \addr \idx}
  \newcommand*{\idxij}{\idxAdd \idx \idxTwo}
  By induction on the structure of \( \tau \).
  The base case where \( \ty = \Tint \nu \form \) is trivial because the definition of \( \SATv H R \val {\Tint \nu \form} \) does not depend on the heap.

  The inductive case is the case where \( \ty = \Tref \idx {\ty'} r \).
  From the definition of \( \mathbf{SATv} \), \( \val \) must be an address \( \raddr \) and must be in \( \dom(H) \).
  We also have
  \begin{gather}
    \Addr \addr \idxij \in \dom(H) \label{eq:lem:heap-extension-preserves-SATv:valid-address} \\
    \SATv H R {H(\AddrArg \addr \idxij)} {[\idxTwo /\idx] \ty'} \label{eq:lem:heap-extension-preserves-SATv:assumption-for-ind}
  \end{gather}
  for each \( \idxTwo \) such that \( r(\idxTwo) > 0 \).
  By the induction hypothesis and~\eqref{eq:lem:heap-extension-preserves-SATv:assumption-for-ind}, we have
  \begin{gather}
    \SATv \HTwo R {H(\AddrArg \addr \idxij)} {[\idxTwo /\idx] \ty'}. \label{eq:lem:heap-extension-preserves-SATv:SATv-HTwo}
  \end{gather}
  Observe that \( \Addr \addr \idxij \neq \Addr {\addr'} {\idx'} \) because of~\eqref{eq:lem:heap-extension-preserves-SATv:valid-address}, and thus \( H(\AddrArg \addr \idxij) = \HTwo(\AddrArg \addr \idxij )\).
  From this,~\eqref{eq:lem:heap-extension-preserves-SATv:valid-address} and~\eqref{eq:lem:heap-extension-preserves-SATv:SATv-HTwo}, we conclude \( \SATv \HTwo R \val \ty \).
\end{proof}

We also show that \( \mathbf{SATv} \) is preserved under the extension of register files.
\begin{definition}
  We write \(R \sqsubseteq R'\) to denote two register files such that:
  \begin{enumerate}
    \item \(dom(R) \subseteq dom(R')\), and
    \item \(\forall x\in dom(R).R(x)=R'(x)\)
  \end{enumerate}
\end{definition}
We say that a register file \( R \) is a \emph{valid substitution} of \( \ty \) if for each \( x \in \FV(\ty) \), \( R(x) \) is an integer.
\begin{lemma}
  \label{lem:register-extension-SATv}
  Suppose that \(\SATv{H}{R}{\val}{\ty}\) and \(R \) is a valid substitution of \( \ty \).
  Then for any \(R'\) such that \(R \sqsubseteq R'\), \(\SATv{H}{R'}{\val}{\ty}\).
\end{lemma}
\begin{proof}
  By induction on \( \ty \).
  We first consider the base case where \( \ty = \Tint \nu \form \).
  Clearly, \( [R][ \val / \nu ] \form = [R'][\val /\nu] \form \) by the fact that \( R \) is a valid substitution of \( \ty \).
  By \( \SATv H R \val \ty \), we have \( \val \in \mathbb Z \) and \( \models [R] [\val / \nu ]\), and thus we conclude \( \SATv H {R'} \val \ty \).
  The inductive case follows from the induction hypothesis.
\end{proof}

\subsubsection*{Main Proof of Subject Reduction}
Finally, we are ready to prove subject reduction.
Prior to starting the standard inductive argument on the transition rules, we state one last lemma, which is a standard substitution lemma.
\begin{lemma}[Substitution]
  \label{lem:substitution}
  If \( \FE \mid \TE \p e : \ty \to \TE' \) and \( x' \notin \dom(\TE) \), then \( \FE \mid [x' / x]\TE \p [x' / x]e : [x' / x]\ty \to [x'/x]\TE' \).
\end{lemma}
\begin{proof}
  By straightforward induction on the type derivation.
\end{proof}
\SubjectReduction*
\begin{proof}
  \newcommand*{\TETwo}{\TE''}
  \newcommand*{\raddr}{\Addr \addr \idx}
  \newcommand*{\idxij}{\idxAdd \idx \idxTwo}

  By induction on the derivation of the transition relation with a case analysis on the last rule applied.

  Before proceeding to the case analysis, let us explain the overall structure that are common across each case.
  By the definition of well-typed configuration, there exists a type environment \( \TE \) that satisfies
  \begin{gather}
    \FE \mid \TE \p  e : \ty \To \TE' \nonumber \\
    \forall \raddr \in \dom(H).\; \Own H R \TE \le 1  \qquad \SAT H R \TE \label{eq:lem:subject-reduction:Own-assumption}
  \end{gather}
  for some \( \FE \).
  Our goal is to find a type environment \( \TETwo \) such that
  \begin{gather*}
    \FE \mid \TETwo \p  e': \ty \To \TE' \nonumber \\
  \forall \raddr \in \dom(H').\; \Own {H'} {R'} \TETwo(a) \le 1 \qquad \SAT {H'} {R'} \TETwo
  \end{gather*}

  For simplicity, we will ignore the subsumption rule \( \rname{T-Sub} \) and argue as if the shape of \( e \) determines the last rule applied to derive \( \FE \mid \TE \p  e : \ty \To \TE' \).
  This does not cause any problem as the conditions about \( \mathbf{Own} \) and \( \mathbf{SAT} \) are preserved by the subtyping relation (Lemma~\ref{lem:subty-preserves-Own-and-SAT}).

  Now we proceed to the case analysis.
  The order of the cases does not follow how the rules are ordered in Figure~\ref{fig:lang}.
  We rather start from interesting cases, namely those that involves heap manipulations or ownership distributions.
  \begin{description}
    \item[Case \rname{R-Assign}:]
      \bgroup %
      \newcommand*{\raddrp}{\Addr {\addr'} {\idx'}}
      \newcommand*{\raddrpArg}{\AddrArg {\addr'} {\idx'}}
      In this case, we must have
      \begin{gather*}
        e = \assnExp x y {e'} \qquad H' = \Hupd H {R(y)} {\Addr \addr \idx}  \qquad R' = R \nonumber \\
        R(x) = \Addr \addr \idx \quad \Addr \addr \idx \in\dom(H)
      \end{gather*}
      for some \( x \), \( y \), \( e' \), \( \addr \) and \( \idx \).
      By inversion of the typing of \( e \), we must also have
      \begin{gather}
        \TE  [x:\Tref{i}{\ty_x'}{r}][y:\ty_1+\ty_2] \nonumber \\
        r(0) = 1 \nonumber \\
        \FE\mid\TE [x\update\Tref{i}{\ty_x}{r}][y\update\ty_1] \p e' :\ty\To\TE' \label{eq:lem:subject-reduction:assign:exp-typable} \\
        i : \Tint \nu {\nu = 0}  \p \ty_x' \approx \ty_2\land_x(x=_{\ty_2}y) \qquad i : \Tint \nu {\nu \neq 0} \p \ty_x' \approx \ty_x \label{eq:lem:subject-reduction:assign:ty-equality-pre}
      \end{gather}
      for some \( r \), \( \ty_x \), \( \ty_x' \), \( \ty_1 \) and \( \ty_2 \).
      Observe that we have
      \begin{align}
         \empEnv \p [0 / \idx ]\ty_x' \approx [0/\idx](\ty_2\land_x(x=_{\ty_2}y)) \qquad \empEnv \p [n / \idx]\ty_x' \approx [n / \idx]\ty_x \quad (n \neq 0) \label{eq:lem:subject-reduction:assign:ty-equality}
      \end{align}
      from~\eqref{eq:lem:subject-reduction:assign:ty-equality-pre}.
      Note that we also have \( H \approx_{\Addr \addr \idx } H' \).

      We set \( \TETwo \defeq \TE[x\update\Tref{i}{\ty_x}{r}][y\update\ty_1] \).
      The typability of \( e' \) under \( \TETwo \) is exactly what~\eqref{eq:lem:subject-reduction:assign:exp-typable} claims.

      We check that \( \forall \raddrp  \in \dom(H') .\ \Own {H'} R \TETwo (\raddrpArg) \le 1 \).
      First, observe that \( x \) has the full ownership of \( \Addr \addr \idx \) because \( r(0) = 1\).
      Formally, we have
      \begin{equation}
        \own H {R(x)} {\TE(x)} (\AddrArg \addr \idx) = 1 \label{eq:lem:subject-reduction:assign:full-ownership}
      \end{equation}
      by the definition of \( \mathbf{own} \).
      Now we would like to appeal to Lemma~\ref{lem:heap-upd-own-preservation}, which states that equivalent heaps have the same \( \mathbf{own} \) function if ``the modified part of the heap is not owned''.
      To this end, we show that
      \begin{align}
        &\text{for each \(  z \in \dom(\TE) \setminus \{ x \} \), } \; \own H {R(z)} {\TE(z)} (\AddrArg \addr \idx) = 0 \label{eq:lem:subject-reduction:assign:own-zero} \\
        &\text{for each \(  j \neq 0 \) with \( r(j) > 0 \),} \; \own H {H(\AddrArg \addr {\idxAdd \idx \idxTwo})} {[\idxTwo / \idx]\ty_x} (\AddrArg \addr \idx) = 0 \label{eq:lem:subject-reduction:assign:own-zeroTwo}
      \end{align}
      The first equality follows from the following inequality:
      \begin{align}
        \Own H R \TE (\AddrArg \addr \idx)
        &= \sum_{z \in \dom(\TE)} \own H {R(z)} {\TE(z)} (\AddrArg \addr \idx)  \tag{by def.} \\
        &= 1 + \sum_{z \in \dom(\TE) \setminus \{ x \} } \own H {R(z)} {\TE(z)} (\AddrArg \addr \idx) \tag{by~\eqref{eq:lem:subject-reduction:assign:full-ownership}}\\
        &\le 1 \tag{by~\eqref{eq:lem:subject-reduction:Own-assumption}}
      \end{align}
      The second equality, i.e.~\eqref{eq:lem:subject-reduction:assign:own-zeroTwo}, can be shown by a similar argument.
      We therefore have:
      \begin{align*}
        &\Own {H'} {R} {\TETwo} \\
        &= \sum_{z \in \dom(\TETwo)} \own {H'} {R(z)} {\TETwo(z)}  \tag{by def.} \\
        &= \own {H'} {R(x)} {\Tref \idx {\ty_x} r} + \own {H'} {R(y)} {\ty_1} \nonumber \\
        &\phantom{=}\quad +  \sum_{z \in \dom(\TETwo) \setminus \{x, y\}} \own H {R(z)} {\TE(z)}
        \tag{by Lemma~\ref{lem:heap-upd-own-preservation} with~\eqref{eq:lem:subject-reduction:assign:own-zero}  and \( \TE(z) = \TETwo(z)\)}
      \end{align*}
    To give a bound to \( \Own {H'} {R} {\TETwo} \), we show that
      \begin{align*}
        &\own {H'} {R(x)} {\Tref \idx {\ty_x} r} + \own {H'} {R(y)} {\ty_1} \\
        &\le \own H {R(x)} {\Tref \idx {\ty'} r} + \own H {R(y)} {\ty_1 + \ty_2}.
      \end{align*}
      because this would imply \( \Own {H'} R \TETwo \le \Own H R {\TE} \).
      For all \( \raddrp \neq \raddr \), we have
      \begin{align*}
       &\own {H'} {R(x)} {\Tref \idx {\ty_x} r}(\raddrpArg) + \own {H'} {R(y)} {\ty_1}(\raddrpArg) \\
       &=\own {H'} \raddr {\Tref \idx {\ty_x} r}(\raddrpArg) + \own {H} {R(y)} {\ty_1}(\raddrpArg) \tag{Lemma~\ref{lem:heap-upd-own-preservation} and \( H \approx_{\Addr \addr \idx} H'\)} \\
       &= \{ \Addr \addr \idxij \mapsto r(\idxTwo) \mid \idxTwo \in \Idx \}(\raddrpArg) \\
       &\phantom{=}+ \sum_{\idxTwo} \own {H'} {H'(\AddrArg \addr \idxij)} {[\idxTwo / \idx]\ty_x}(\raddrpArg) + \own {H} {R(y)} {\ty_1}(\raddrpArg) \tag{by def. of \( \mathbf{own}\)} \\
       &= \{ \Addr \addr \idxij \mapsto r(\idxTwo) \mid \idxTwo \in \Idx \}(\AddrArg {\addr'} {\idx'}) + \sum_{\idxTwo: \idxTwo \neq 0} \own {H'} {H(\AddrArg \addr \idxij)} {[\idxTwo / \idx]\ty_x}(\raddrpArg)\\
       &\phantom{=}+ \own H {R(y)} {[0 / \idx]\ty_x} (\AddrArg {\addr'} {\idx'}) + \own {H} {R(y)} {\ty_1}(\raddrpArg) \tag{by \( H'(\AddrArg \addr \idx) = R(y)\) and \( H \approx_{\raddr} H' \)} \\
       &= \{ \Addr \addr \idxij \mapsto r(\idxTwo) \mid \idxTwo \in \Idx \}(\raddrpArg) + \sum_{\idxTwo: \idxTwo \neq 0} \own {H'} {H(\AddrArg \addr \idxij)} {[\idxTwo / \idx]\ty'}(\AddrArg {\addr'} {\idx'})\\
       &\phantom{=}+ \own H {R(y)} {\ty_2} (\AddrArg {\addr'} {\idx'}) + \own {H} {R(y)} {\ty_1}(\raddrpArg) \tag{by~\eqref{eq:lem:subject-reduction:assign:ty-equality} and Lemma~\ref{lem:strengthning-preserves-own}} \\
       &= \{ \Addr \addr \idxij \mapsto r(\idx) \mid \idx \in \Idx \}(\raddrpArg) + \sum_{\idxTwo: \idxTwo \neq 0} \own {H'} {H(\AddrArg \addr \idxij)} {[\idxTwo / \idx]\ty'}(\raddrpArg)\\
       &\phantom{=}+ \own {H} {R(y)} {\ty_1 + \ty_2}(\raddrpArg) \tag{by Lemma~\ref{lem:own-ty-add}} \\
       &= \{ \Addr \addr \idxij \mapsto r(\idxTwo) \mid \idxTwo \in \Idx \}(\AddrArg {\addr'} {\idx'}) + \sum_{\idxTwo: \idxTwo \neq 0} \own {H} {H(\AddrArg \addr \idxij)} {[\idxTwo / \idx]\ty'}(\AddrArg {\addr'} {\idx'})\\
       &\phantom{=}+ \own {H} {R(y)} {\ty_1 + \ty_2}(\raddrpArg) \tag{Lemma~\ref{lem:heap-upd-own-preservation} and \( H \approx_{\Addr \addr \idx} H'\)}  \\
       &\le \{ \Addr \addr \idxij \mapsto r(\idxTwo) \mid \idxTwo \in \Idx \}(\raddrpArg) + \sum_{\idxTwo: \idxTwo \neq 0} \own {H} {H(\AddrArg \addr \idxij)} {[\idxTwo / \idx]\ty'}(\raddrpArg)\\
       &\phantom{le}+ \own H {H(\AddrArg \addr \idx)} {\ty'} (\raddrpArg) + \own {H} {R(y)} {\ty_1 + \ty_2}(\raddrpArg) \\
       &= \own H {R(x)} {\Tref \idx {\ty'_x} r}(\raddrpArg) + \own H {R(y)} {\ty_1 + \ty_2}(\raddrpArg).
      \end{align*}
      The case where the argument is \( \raddr \) can be proved in a similar manner using Lemma~\ref{lem:heap-upd-own-preservation} and Lemma~\ref{lem:own-ty-add}.

      Now we show that \( \SAT {H'} R \TETwo \).
      By \( \SAT H R \TE \), for each \( z \in \dom( \TE ) \), we have \( \SATv H R {R(z)} {\TE(z)} \).
      Because satisfiability is preserved between equivalent heaps, we have \( \SATv {H'} R {R(z)} {\TE''(z)} \) for \( z \in \dom(\TE) \setminus \{ x, y \}\); more formally, this is by Lemma~\ref{lem:hp-upd-SATv-preservation} with \( H \approx_{\raddr} H' \) and~\eqref{eq:lem:subject-reduction:assign:own-zero}.
      Similarly, we have \( \SATv {H'} R {R(y)} {\ty_1 + \ty_2}\).
      In particular, we have
      \begin{align}
        \SATv {H'} R {R(y)} {\ty_1} \quad \text{and} \quad \SATv {H'} R {R(y)} {\ty_2} \label{eq:lem:subject-reduction:assign:satv-y-add}
      \end{align}
      by Lemma~\ref{lem:SATv-ty-add}.
      Since \( \ty_1 = \TETwo(y) \), it remains to show \( \SATv {H'} R {R(x)} {\TETwo(x)} \).
      We need to show that \( \SATv {H'} R {H'(\AddrArg \addr \idxij)} {[\idxTwo / \idx] \ty'_x} \) for each \( \idxTwo\) such that \( r(\idxTwo) > 0\).
      For \( j \neq 0 \), we can appeal to Lemma~\ref{lem:hp-upd-SATv-preservation} together with~\eqref{eq:lem:subject-reduction:assign:own-zeroTwo} and the fact that \( \SATv {H} R {R(x)} {\TE(x)} \) with \( \empEnv \vdash [\idxTwo / \idx]\ty_x' \approx [\idxTwo/ \idx] \ty_x \) and Lemma~\ref{lem:ty-equiv-preserves-own-and-SATv}.
      We also have \( \SATv {H'} R {H'(\AddrArg \addr \idx)} {\ty_x'[0 / \idx ]} \) from~\eqref{eq:lem:subject-reduction:assign:satv-y-add} because  \( H'(\AddrArg \addr \idx) = R(y)\) (hence the strengthening has no effect) and~\eqref{eq:lem:subject-reduction:assign:ty-equality}.
      \egroup

    \item[Case \rname{R-Deref}:]
      We must have
      \begin{gather}
        e = \letexp{x}{*y}{e_0} \quad e' = [x' / x]e_0 \quad H = H' \nonumber \\
        R' = \Rupd R {x'} {H(R(y))} \quad x' \notin \dom(H) \nonumber
      \end{gather}
      for some \( x \), \( x' \), \( y \) and \( e_0\).
      Furthermore, by inversion of the typing of \( e \), we have
      \begin{gather}
        \TE[y : \Tref \idx {\ty_y}{r}] \nonumber \\
        \FE\mid\TE[y\update\Tref{i}{\ty_y'}{r}],x:\ty_x\p e:\ty\To\TE' \label{eq:lem:subject-reduction:deref:exp-typable} \\
        i : \Tint \nu {\nu = 0}  \p \ty' + \ty_x \approx \ty_y \qquad i : \Tint \nu {\nu = 0}  \p \ty_y' \approx \ty'\land_y(y=_{\ty'}x) \label{eq:lem:subject-reduction:deref:ty-equality-zero}\\
        i : \Tint \nu {\nu \neq 0}  \p \ty_y' \approx \ty_y \label{eq:lem:subject-reduction:deref:ty-equality-nonzero}\\
        x\notin\dom(\TE') \qquad r(0)>0 \nonumber
      \end{gather}
      for some  \( \ty_x \), \( \ty_y \), \( \ty_y' \) \( \ty' \) and \( r \).

      We take \( \TE[y\update\Tref{i}{\ty_y'}{r}], x' : \ty_x\ \) for \( \TETwo \).
      Then from the substitution lemma (Lemma~\ref{lem:substitution}),~\eqref{eq:lem:subject-reduction:deref:exp-typable} and \( x \notin \TE' \), we have \( \FE \mid \TETwo \p e' : \ty \To \TE' \).

      We now show that  \( \forall \Addr \addr \idx. \Own H {R'} {\TETwo} (\AddrArg \addr \idx) \le 1\).
      Since the differences between \( \TE \) and \( \TETwo \) are in the type of \( x' \) and \( y \), we only need to show that
      \begin{align*}
        &\own H {R(y)} {\Tref \idx {\ty_y}{r}} \\
        &= \own[R'] H {R'(y)} {\Tref \idx {\ty'_y}{r}} {} + \own[R'] H {R'(x')} {\ty_x}.
      \end{align*}
      Since \( R(y)\) must be an address, let us write \( \Addr \addr \idx \) for \( R(y) \).
      The above equation is shown as follows:
      \begin{align*}
        & \own H {R(y)} {\Tref \idx {\ty_y}{r}} \\
        &= \{ \Addr \addr \idxij \mapsto r(\idxTwo) \mid \idxTwo \in \Idx \} + \sum_\idxTwo \own{H} {H(\AddrArg \addr \idxij)} {[\idxTwo / \idx]\ty_y} \tag{by def.} \\
        &= \{ \Addr \addr \idxij \mapsto r(\idxTwo) \mid \idxTwo \in \Idx \} + \sum_{\idxTwo: \idxTwo \neq 0} \own{H} {H(\AddrArg \addr \idxij)} {[\idxTwo / \idx]\ty_y} \\
        &\phantom{=} +  \own{H} {H(\AddrArg \addr \idx)} {[0 / \idx]\ty_y} \\
        &= \{ \Addr \addr \idxij \mapsto r(\idxTwo) \mid \idxTwo \in \Idx \} + \sum_{\idxTwo: \idxTwo \neq 0} \own{H} {H(\AddrArg \addr \idxij)} {[\idxTwo / \idx]\ty_y'} \\
        &\phantom{=} +  \own{H} {H(\AddrArg \addr \idx)} {\ty_x + \ty_y'[0 / \idx]} \tag{by Lemmas~\ref{lem:ty-equiv-preserves-own-and-SATv} and~\ref{lem:strengthning-preserves-own} with~\eqref{eq:lem:subject-reduction:deref:ty-equality-zero} and~\eqref{eq:lem:subject-reduction:deref:ty-equality-nonzero}} \\
        &= \{ \Addr \addr \idxij \mapsto r(\idxTwo) \mid \idxTwo \in \Idx \} + \sum_{\idxTwo: \idxTwo \neq 0} \own{H} {H(\AddrArg \addr \idxij)} {[\idxTwo / \idx]\ty'_y} \\
        &\phantom{=} +  \own{H} {H(\AddrArg \addr \idx)} {\ty_x} +  \own{H} {H(\AddrArg \addr \idx)} {\ty'_y[0 / \idx]}  \tag{Lemma~\ref{lem:own-ty-add}}\\
        &= \{ \Addr \addr \idxij \mapsto r(\idxTwo) \mid \idxTwo \in \Idx \} + \sum_{\idxTwo} \own{H} {H(\AddrArg \addr \idxij)} {[\idxTwo / \idx]\ty'_y} \\
        &\phantom{=} +  \own{H} {R'(x')} {\ty_x}   \tag{\( R'(x') = H(R(y)) \)}\\
        &= \own[R']{H} {R'(y)} {\Tref \idx {\ty'_y} r} +  \own[R']{H} {R'(x')} {\ty_x}
      \end{align*}

      For this case, it remains to show that \( \SAT H {R'} \TETwo \).
      Again, it suffices to consider \( \SATv H {R'} {R'(x')} {\TETwo(x')} \) and \( \SATv H {R'} {R'(y)} {\TETwo(y)} \).
      Since \( \SAT H R \TE \), we have \( \SATv H {R} {R(y)} {\TE(y)} \), and thus
      \begin{align*}
        \SATv H R {H(\AddrArg \addr {\idxAdd \idx \idxTwo})} {[\idxTwo/ \idx] \ty_y}
      \end{align*}
      for each \( r \) such that \( r(\idxTwo) > 0 \), where \( \Addr \addr \idx = R(y) \).
      In particular, since \( r(0) > 0 \) and~\eqref{eq:lem:subject-reduction:deref:ty-equality-zero}, we have
      \begin{align*}
        \SATv H R {H(R(y))} {\ty_x} \quad \text{and} \quad \SATv H R {H(R(y))} {[0 / \idx]\ty'}
      \end{align*}
      thanks to Lemma~\ref{lem:SATv-ty-add}.
      Because \( R'(x') = H(R(y)) \) and \( R' \) is an extension of \( R \), we have \( \SATv H {R'} {R'(x')} {\TETwo(x')}\) by Lemma~\ref{lem:register-extension-SATv}.
      Similarly, we also have \( \SATv H {R'} {R'(y)} {\TETwo(y)} \).

    \item[Case \rname{R-AddPtr}:]
      \bgroup %
      \newcommand*{\idxiz}{\idxAdd \idx {R(z)}}
      In this case, we must have
      \begin{gather}
        e = \letexp{x}{y\pplus z}{e_0} \qquad e' = [x'/x]e_0 \qquad H' = H   \nonumber \\
        R' = \Rupd R {x'} {\Addr \addr \idxiz} \qquad R(y) = \Addr \addr \idx \qquad x' \notin \dom(R)
      \end{gather}
      for some \( x \), \( y \), \( z\), \( e_0 \) and \( \Addr \addr \idx \).
      By inversion of the typing of \( e \), it must also be the case that
      \begin{gather}
        \TE[y:\Tref{i}{(\ty_1+\ty_2)}{r_y}][z:\Tint{\nu}{\form}]\nonumber \\
        \FE\mid\TE[y\update\Tref{i}{\ty_1}{r_y'}],x:\Tref{i}{[(i + z)/i]\ty_2}{r_x},z:\Tint{\nu}{\form}\p e_0 :\ty\To\TE' \label{eq:lem:subject-reduction:addptr:exp-typable} \\
        x \notin \dom(\TE') \nonumber \\
        r_x(\idx - R(z))+r_y'(\idx)=r_y(\idx) \label{eq:lem:subject-reduction:addptr:r}
      \end{gather}
      for some \( r_x \), \( r_y\), \( r_y'\), \( \ty_1 \) and \( \ty_2 \).

      We define \( \TETwo \defeq \TE[y\update\Tref{i}{\ty_1}{r_y'}], x' :\Tref{i}{[(i-z)/i]\ty_2}{r_x},z:\Tint{\nu}{\form} \).
      By substitution lemma (Lemma~\ref{lem:substitution}), \eqref{eq:lem:subject-reduction:addptr:exp-typable} and \( x \notin \dom(\TE') \) we have
      \( \TETwo \p e' : \ty \To \TE' \).

      We check the condition \( \forall \Addr {\addr'} {\idx'}. \Own H {R'} {\TETwo} (\AddrArg {\addr'} {\idx'}) \le 1\).
      It is enough to show that
      \begin{align*}
        &\own H {R(y)} {\Tref \idx {(\ty_1+\ty_2)} {r_y}} \\
        &=  \own[R'] H {R'(x')} {\Tref \idx {[\idxAdd \idx z/\idx]\ty_2}{r_x}} + \own[R'] H {R'(y)} {\Tref \idx {\ty_1}{r_y'}}.
      \end{align*}
      because then we have \( \Own H R \TE = \Own H {R'} {\TETwo} \).
      Indeed, we have
      \begin{align*}
        &\own H {R(y)} {\Tref \idx {(\ty_1+\ty_2)} {r_y}} \\
        &= \{ \Addr \addr \idxij \mapsto r_y(\idxTwo) \mid \idxTwo \in \Idx \} + \sum_{\idxTwo} \own H {\Addr \addr \idxij} {[\idxTwo / \idx] \ty_1 + \idx[\idxTwo / \idx]\ty_2 }\tag{by def. and \( R(y) = \Addr \addr \idx \)} \\
        &= \{ \Addr \addr \idxij \mapsto r_y(\idxTwo)  \mid \idxTwo \in \Idx \} + \sum_{\idxTwo} \own H {\Addr \addr \idxij} {[\idxTwo / \idx] \ty_1} \\
        &\phantom{=} + \sum_{\idxTwo} \own H {\Addr \addr \idxij} {[\idxTwo / \idx]\ty_2 } \tag{Lemma~\ref{lem:own-ty-add}} \\
        &= \{ \Addr \addr \idxij \mapsto r_x(\idxTwo - {R(z)}) + r'_y(\idxTwo)  \mid \idxTwo \in \Idx \} + \sum_{\idxTwo} \own H {\Addr \addr \idxij} {[\idxTwo / \idx] \ty_1} \\
        &\phantom{=} + \sum_{\idxTwo} \own H {\Addr \addr \idxij} {[\idxTwo / \idx]\ty_2 } \tag{by~\eqref{eq:lem:subject-reduction:addptr:r}}\\
        &= \{ \Addr \addr \idxij \mapsto  r'_y(\idxTwo)  \mid \idxTwo \in \Idx \} + \sum_{\idxTwo} \own H {\Addr \addr \idxij} {[\idxTwo / \idx] \ty_1} \\
        &\phantom{=} +   \{ \Addr \addr \idxij \mapsto  r_x(\idxTwo-R(z))  \mid \idxTwo \in \Idx \}  + \sum_{\idxTwo} \own H {\Addr \addr \idxij} {[\idxTwo / \idx]\ty_2 } \\
        &= \own[R'] H {R'(y)} {\Tref \idx {\ty_1} {r'_y}}\\
        &\phantom{=} +   \{ \Addr \addr {\idxAdd \idxiz \idxTwo} \mapsto  r_x(\idxTwo)  \mid \idxTwo \in \Idx \}  + \sum_{\idxTwo} \own H {\Addr \addr {\idxAdd \idxiz \idxTwo}} {[\idxAdd {R(z)} \idxTwo / \idx]\ty_2 } \\
        &= \own H {R'(y)} {\Tref \idx {\ty_1} {r'_y}}\\
        &\phantom{=} +   \{ \Addr \addr {\idxAdd \idxiz \idxTwo} \mapsto  r_x(\idxTwo)  \mid \idxTwo \in \Idx \}  + \sum_{\idxTwo} \own H {\Addr \addr {\idxAdd \idxiz \idxTwo}} {[\idxTwo / \idx][\idxAdd \idx z/ \idx]\ty_2 } \\
        &= \own[R'] H {R'(y)} {\Tref \idx {\ty_1} {r'_y}} + \own[R'] H {R'(x')} {\Tref \idx {[\idx + z / \idx ]\ty_2} {r_x}} \tag{\( R'(x') = \Addr \addr \idxiz \)}
      \end{align*}

      Similarly, to show that \( \SAT H {R'} \TETwo \) it suffices to show that, \( \SATv H {R'} {R'(x')} {\TETwo(x')} \) and \( \SATv H {R'} {R'(x')} {\TETwo(y)} \).
      Since \( \SAT H R \TE \), we have \( \SATv H R {R(y)} {\TE(y)} \).
      In particular, we have \( \SATv H R {H(\AddrArg \addr \idxij)} {[\idxTwo / \idx] \tau_1 + [\idxTwo / \idx] \ty_2}\) for  each \( \idxTwo \) such that \( r_y(\idxTwo) > 0 \).
      Therefore,
      \begin{align*}
        \SATv H R {H(\AddrArg \addr \idxij)} {[\idxTwo / \idx] \tau_1} \quad \text{and} \quad  \SATv H R {H(\AddrArg \addr \idxij)} {[\idxTwo / \idx] \ty_2}
      \end{align*}
      for  each \( \idxTwo \) such that \( r_y(\idxTwo) > 0 \), by Lemma~\ref{lem:SATv-ty-add}.
      Since \( r_y(\idxTwo) \ge r'_y (\idxTwo) \) and \( r_y(\idxTwo) \ge r_x (\idxTwo - R(z)) \) by~\eqref{eq:lem:subject-reduction:addptr:r}, we have
      \begin{align*}
        \SATv H {R'} {H(\AddrArg \addr \idxij)} {[\idxTwo / \idx] \tau_1} \quad (\text{for \( \idxTwo\) s.t. \( r'_y(\idxTwo) > 0\)}) \\
        \SATv H {R'} {H(\AddrArg \addr {\idxAdd \idxiz \idxTwo})} {[\idxTwo / \idx][\idxAdd \idx z/ \idx] \tau_2} \quad (\text{for \( \idxTwo\) s.t. \( r_x(\idxTwo) > 0\)})
      \end{align*}
      with the help of Lemma~\ref{lem:register-extension-SATv}.
      Hence, we conclude \( \SATv H {R'} {R'(y)} {\TETwo(y)} \) and \( \SATv H {R'} {R'(x')} {\TETwo(x')} \) as desired.

      \egroup

    \item[Case \rname{R-AliasAddPtr}]
      \bgroup %
      \newcommand*{\xidx}{\idxAdd \idx R(z)}
      \newcommand*{\xaddr}{\Addr \addr \xidx}
      It must be the case that
      \begin{gather*}
        e = \aliasexp{x}{y\pplus z}e' \\
        H' = H \qquad R' = R \qquad R(x) = \xaddr \qquad R(y) = \raddr
      \end{gather*}
      for some \( x \), \( y \) and \( z \).
      By inversion on the typing of \( e \) we must also have
      \begin{gather}
        \TE[x:\Tref{i}{\ty_x}{r_x}][y:\Tref{i}{\ty_y}{r_y}][z:\Tint{\nu}{\form}] \nonumber \\
        \FE \mid \TE[x\update\Tref{i}{\ty_x'}{r_x'}][y\update\Tref{i}{\ty_y'}{r_y'}],z:\Tint{\nu}{\form}\p e':\ty\To\TE' \nonumber \\
        \empEnv \p (\Tref{i}{[(i - z)/i]\ty_x}{r_{x_z}} + \Tref{i}{\ty_y}{r_y}) \approx (\Tref{i}{[(i-z)/i]\ty_x'}{r_{x_z}'} + \Tref{i}{\ty_y'}{r_y'}) \label{eq:lem:subject-reduction:aliasaddptr:ty-equiv} \\
         r_{x_z}(i)=r_x(i-R(z))\andalso r_{x_z}'(i)=r_x'(i-R(z)) \label{eq:lem:subject-reduction:aliasaddptr:own-fun}
      \end{gather}

      We define \( \TETwo \) as \( \TE[x\update\Tref{i}{\ty_x'}{r_x'}][y\update\Tref{i}{\ty_y'}{r_y'}],z:\Tint{\nu}{\form} \).
      Obviously, \( \FE \mid \TETwo \p e' : \ty' \To \TE' \).

      Now we check that the invariant about the ownership is preserved.
      It is enough to check that
      \begin{align*}
        &\own H {R(x)} {\TE(x)} + \own H {R(y)} {\TE(y)} \\
        &= \own H {R(x)} {\TETwo(x)} + \own H {R(y)} {\TETwo(y)}.
      \end{align*}
      Observe that
      \begin{align*}
        &\own H {R(x)} {\TE (x)} \\
        &= \own H \xaddr {\Tref{i}{\ty_x}{r_x}} \\
        &= \{ \Addr \addr {\idxAdd {\idxAdd \idx R(z)} \idxTwo} \mapsto r_x(\idxTwo) \mid \idxTwo \in \Idx \} + \sum_\idxTwo \own H {\Addr \addr {\idxAdd {\idxAdd \idx R(z)} \idxTwo}} {[j /\idx]\ty_x} \\
        &= \{ \Addr \addr {\idxAdd \idx \idxTwo} \mapsto r_x(\idxTwo - R(z)) \mid \idxTwo \in \Idx \} + \sum_\idxTwo \own H {\Addr \addr {\idxAdd \idx \idxTwo}} {[j - R(z) /\idx]\ty_x} \\
        &= \{ \Addr \addr {\idxAdd \idx \idxTwo} \mapsto r_{x_z}(\idxTwo) \mid \idxTwo \in \Idx \} + \sum_\idxTwo \own H {\Addr \addr {\idxAdd \idx \idxTwo}} {[\idxTwo / \idx][\idx - z / \idx]\ty_x} \tag{by~\eqref{eq:lem:subject-reduction:aliasaddptr:own-fun}} \\
        &= \own H {R(y)} {\Tref{i}{[(i - z)/i]\ty_x}{r_{x_z}}} \tag{by \( R(y) = \addr \)}
      \end{align*}
      Similarly, we have \( \own H {R(x)} {\TETwo(x)} = \own H {R(y)} {\Tref{i}{[(i-z)/i]\ty_x'}{r_{x_z}'}} \).
      Hence, we have
      \begin{align*}
        &\own H {R(x)} {\TE(x)} + \own H {R(y)} {\TE(y)} \\
        &= \own H {R(y)} {\Tref{i}{[(i - z)/i]\ty_x}{r_{x_z}}} + \own H {R(y)} {\Tref{\idx}{\ty_y}{r_y}} \\
        &= \own H {R(y)} {\Tref{i}{[(i - z)/i]\ty_x}{r_{x_z}} + \Tref{\idx}{\ty_y}{r_y}} \tag{Lemma~\ref{lem:own-ty-add}} \\
        &= \own H {R(y)} {\Tref{i}{[(i - z)/i]\ty'_x}{r_{x_z}'} + \Tref{\idx}{\ty'_y}{r'_y}} \tag{Lemma~\ref{lem:ty-equiv-preserves-own-and-SATv} with~\eqref{eq:lem:subject-reduction:aliasaddptr:ty-equiv}} \\
        &= \own H {R(y)} {\Tref{i}{[(i-z)/i]\ty_x'}{r_{x_z}'}} + \own H {R(y)} {\Tref{\idx}{\ty_y'}{r_y'}} \tag{Lemma~\ref{lem:own-ty-add}} \\
        &= \own H {R(x)} {\TETwo(x)} + \own H {R(y)} {\TETwo(y)}.
      \end{align*}
      as desired.

      To finish this case we show that \( \SAT H R \TETwo \).
      Once again, we may focus on \( x \) and \( y \).
      By a similar reasoning we did above for \( \mathbf{own} \), we have
      \begin{align*}
        \SATv H R {R(x)} {\TE(x)} \quad &\text{iff} \quad \SATv H R {R(y)} {\Tref{i}{[(i-z)/i]\ty_x}{r_{x_z}}} \\
        \SATv H R {R(x)} {\TETwo(x)} \quad &\text{iff} \quad \SATv H R {R(y)} {\Tref{i}{[(i-z)/i]\ty_x'}{r_{x_z}'}}
      \end{align*}
      Since \( \SAT H R \TE \), we have \( \SATv H  R {R(x)} {\TE(x)} \) and \( \SATv H R {R(y)} {\TE(y)} \).
      Hence, from the above bi-implication and Lemma~\ref{lem:SATv-ty-add}, it follows that \( \SATv H R {R(y)} {\Tref{i}{\ty_y}{r_y}] + \Tref{i}{[(i-z)/i]\ty_x}{r_{x_z}}}\).
      Lemma~\ref{lem:ty-equiv-preserves-own-and-SATv} with~\eqref{eq:lem:subject-reduction:aliasaddptr:ty-equiv}
      gives us \( \SATv H R {R(y)} {\Tref{i}{\ty'_y}{r'_y}} + \Tref{i}{[(i-z)/i]\ty'_x}{r'_{x_z}} \).
      By Lemma~\ref{lem:SATv-ty-add} together with the above bi-implication, we conclude \( \SATv H R {R(x)} {\TETwo(x)} \) and \( \SATv H R {R(y)} {\TETwo(y)}\).
      \egroup %
    \item[Case \rname{R-AliasDeref}]
      Similar to the case of \rname{R-AliasAddPtr}.
    \item[Case \rname{R-MkArray}]
      For this case, we have
      \begin{gather*}
        e =  \letexp{x}{\mkarray\;n}e_0 \qquad e' = [x' / x]e_0 \\
        H' = H \set{\Addr \addr 0 \mapsto m_0}\cdots \set{\Addr \addr {n-1} \mapsto m_{n-1}} \qquad \Addr \addr 0, \ldots, \Addr \addr {n - 1} \notin \dom(H) \\
        R' = \Rupd R {x'} {\Addr \addr 0} \qquad x' \notin \dom(R)
      \end{gather*}
      for some \( x \), \( x' \), \( e_0\) and a base address \( a \).
      By inversion on the typing of \( e_0 \), we also have
      \begin{gather}
         \FE\mid\TE,x:\Tref{i}{\top_k}{r}\p e_0:\ty\To\TE' \label{eq:lem:subject-reduction:mkarray:exp-typable}\\
         r(i)=\begin{cases}
                1&0\leqq i\leqq n-1\\
                0&otherwise
              \end{cases} \label{eq:lem:subject-reduction:ownership}\\
        x\notin\dom(\TE') \nonumber
      \end{gather}

      We set \( \TETwo \defeq  \TE,x:\Tref{\idx}{\top_k}{r} \).
      The fact that  \( \FE \mid \TETwo \p e'  : \ty \To \TE' \) is immediate from substitution lemma (Lemma~\ref{lem:substitution}),~\eqref{eq:lem:subject-reduction:mkarray:exp-typable} and \( x \notin \dom(\TE') \).

      We first check that \( \SAT {H'} {R'} \TETwo \).
      Since \( \SAT H R \TE \), we have
      \begin{align}
        \SATv H R {R(z)} {\TE(z)} \qquad \text{for each \( z \in \dom(\TE) \)} \label{eq:lem:subject-reduction:mkarray:SATv}
      \end{align}
      Because \( R \sqsubseteq R' \) and \( \TE(z) = \TETwo(z) \) for \( z \in \dom(\TE)\), we have \( \SATv H R {R'(z)} {\TETwo(z)} \) by Lemma~\ref{lem:register-extension-SATv}.
      Hence, it remains to show that \( \SATv {H'} {R'} {R(x')} {\Tref{i}{\top_k}{r}} \).
      This is immediate by the definition of \( \mathbf{SATv} \) and the fact that ``\( \top_k \) is valid'' (Lemma~\ref{lem:SATv-top}).

      Now we check that, for every \( \Addr {\addr'} {\idx'} \in \dom(H') \), \( \Own {H'} {R'} \TETwo (\AddrArg {\addr'} {\idx'}) \le 1 \).
      Observe that \( \own H {R(z)} {\TE(z)} = \own {H'} {R'(z)} {\TETwo(z)} \) for each \( z \in \dom(\TE) \) thanks to Lemma~\ref{lem:heap-extension-preserves-own} and~\eqref{eq:lem:subject-reduction:mkarray:SATv}.
      We therefore have
      \begin{align*}
        \Own {H'} {R'} \TETwo
        &= \Own H R \TE + \own[R'] {H'} {R'(x')}{\TETwo(x')} \\
        &= \Own H R \TE + \own[R'] H {\Addr \addr 0} {\Tref{\idx}{\top_k}{r}}  \\
        &= \Own H R \TE +  \{ \Addr \addr \idx \mapsto r(\idx) \mid \idx \in \Idx \} \tag{Lemma~\ref{lem:own-top-is-zero}}
      \end{align*}
      Thus, for each \( \Addr {\addr'} {\idx'} \in \dom(H')\) if \( \Addr {\addr'} {\idx'} \in \dom(H) \), then \( \Own {H'}{R'} \TETwo (\AddrArg {\addr'} {\idx'}) = \Own H R \TE (\AddrArg {\addr'} {\idx'}) \le 1 \); otherwise \( \Own {H'}{R'} \TETwo (\AddrArg {\addr'} {\idx'}) =  1 \) by~\eqref{eq:lem:subject-reduction:ownership}.

    \item[Case \rname{R-LetVar}]
      In this case, we must have
      \begin{gather}
        e = \letexp x y {e_0} \qquad e' = [x'/x] e_0 \qquad H' = H  \nonumber \\
        R' = \Rupd R {x'} {R(y)} \qquad x' \notin \dom(R) \nonumber
      \end{gather}
      for some \( x \), \( y \), \( e_0 \), \( \addr \) and \( \idx \).
      By inversion of the typing of \( e \), we must also have
      \begin{gather}
        \FE \mid \TE \p y : \ty_1 \To \TE_1 \label{eq:lem:subject-reduction:letvar:y-typable} \\
        \FE \mid \TE_1, x : \ty_1 \p e_0 : \ty \To \TE' \label{eq:lem:subject-reduction:letvar:exp-typable} \\
        x \notin \dom(\TE') \qquad x \notin \fv \ty \nonumber
      \end{gather}
      for some \( \TE_1 \) and \( \ty_1 \).
      By inversion on~\eqref{eq:lem:subject-reduction:letvar:y-typable}, \( \TE \) and \( \TE_1 \) must satisfy
      \begin{gather*}
        \TE[y : \ty_1 + \ty_2] \qquad \TE_1 = \TE [y \update  \ty_2]
      \end{gather*}
      for some \( \ty_2 \).

      We define \( \TETwo  \defeq \TE [y \update  \ty_2], x' : \ty_1 \).
      By substitution lemma (Lemma~\ref{lem:substitution}),~\eqref{eq:lem:subject-reduction:letvar:exp-typable} and \( x \notin \dom(\TE') \cup \fv{\ty}\), we have \( \FE \mid \TE'' \p [x'/ x]e_0 : \ty \To \TE' \).

      Now we check that the ownership invariant is preserved.
      As in other cases, it suffices to check that
      \begin{align*}
        \own H {R(y)} {\TE(y)} = \own H {R'(y)} {\TETwo(y)} + \own H {R'(x')} {\TETwo(x')}
      \end{align*}
      This trivially holds because
      \begin{align*}
        &\own H {R'(y)} {\TETwo(y)} + \own H {R'(x')} {\TETwo(x')} \\
        &=\own H {R(y)} {\ty_2} + \own H {R(y)} {\ty_1} \tag{by \( R'(y) = R(y) \) and \( R'(x') = R(y)\)} \\
        &=\own H {R(y)} {\ty_1 + \ty_2} \tag{Lemma~\ref{lem:own-ty-add}} \\
        &=\own H {R(y)} {\TE(y)}.
      \end{align*}

      Finally, we show that \( \SAT H {R'} \TETwo \).
      We only need to check \( \SATv H {R'} {R'(y)} {\TETwo(y)} \) and \( \SATv H {R'} {R'(x')} {\TETwo(x')} \).
      By \( \SAT H R \TE \), we have \( \SATv H R {R(y)} {\ty_1 + \ty_2}\).
      Applying Lemma~\ref{lem:SATv-ty-add} and~\ref{lem:register-extension-SATv} together with \( R \sqsubseteq  R' \) and \( \ty_1 + \ty_2 = \TETwo(y) + \TETwo(x') \) concludes this case.
    \item[Case \rname{R-LetInt}]
      Similar to the case for \rname{R-LetVar}.
    \item[Case \rname{R-Context}]
      By induction on the structure of \( E \).
      If \( E = \hole \) then the result holds by the induction hypothesis (of the induction on the transition rules).

      Now suppose that
      \begin{gather}
        e = E[e_0] \qquad e' = E[e_0'] \qquad E = \letexp{x}{E'}e_2 \nonumber \\
        \sconfig R H {e_0} \red \sconfig{R'}{H'}{e_0'} \label{eq:lem:subject-reduction:context:red}
      \end{gather}
      By inversion on the typing of \( e \), we have
      \begin{gather}
        \FE \mid \TE \p E[e_0] : \ty_1 \To \TE_1 \nonumber \\
        \FE \mid \TE_1, x : \ty_1 : e_2 : \tau \To \TE' \label{eq:lem:subject-reduction:context:expTwo-typable}\\
        x \notin  \dom(\TE_1) \cup \FV(\ty_1) \nonumber
      \end{gather}
      for some \( \ty_1 \) and \( \TE_1\).
      We have
      \begin{align*}
        \sconfig R H {E'[e_0]} \red \sconfig{R'}{H'}{E'[e_0']}
      \end{align*}
      by \rname{R-Context} and~\eqref{eq:lem:subject-reduction:context:red}.

      By the induction hypothesis (on the evaluation context) there exists a type environment \( \TETwo \) such that
      \begin{gather}
        \FE \mid \TETwo \p E'[e'_0] : \ty_1 \To \TE_1 \label{eq:lem:subject-reduction:context:exp-after-red-typable} \\
        \forall \raddr . \Own {H'} {R'} \TETwo (\AddrArg \addr \idx) \le 1 \qquad \SAT {H'} {R'} {\TETwo} \nonumber
      \end{gather}
      It remains to check that \( \FE \mid \TETwo \p e' : \ty \To \TE' \).
      This is immediate from~\eqref{eq:lem:subject-reduction:context:expTwo-typable} and~\eqref{eq:lem:subject-reduction:context:exp-after-red-typable}.

    \item[Case \rname{R-AssertFail}]
      This case cannot happen by Lemma~\ref{lem:assert-fail-never-happens}.
    \item[Case \rname{R-Assert}]
      In this case, we must have
      \begin{gather*}
        e = \assertexp \form; e' \quad R' = R \quad H' = H.
      \end{gather*}
      Furthermore, we have
      \begin{gather*}
        \FE \mid \TE \p e' : \ty \To \TE'
      \end{gather*}
      by inversion on the typing of \( e \).
      Therefore, we can simply take \( \TE \) for \( \TETwo \).
    \item[Case \rname{R-IfTrue}]
      It must be the case that
      \begin{gather*}
        e = \ifexp x {e_1} {e_2} \qquad e' = e_1 \qquad H' = H \\
        R' = R \qquad R(x) \le 0
      \end{gather*}
      By inversion on the typing of \( e \), we have
      \begin{align}
        &\TE[x : \Tint \nu \form ] \nonumber \\
        &\FE \mid \TE[x \update \Tint \nu {\form \land \nu \le 0} ] \p e_1 : \ty \To \TE' \label{eq:lem:subject-reduction:iftrue:exp-typable}
      \end{align}

      We set \( \TETwo \defeq \TE[x \update \Tint \nu {\form \land \nu \le 0} ] \).
      Clearly, we have \( \FE \mid \TETwo \p e' : \ty \To \TE' \) from~\eqref{eq:lem:subject-reduction:iftrue:exp-typable},
      Since \( \own H {R(x)} {\TE(x)} = \zerofun =  \own H {R(x)} {\TETwo(x)} \), we have \( \Own H R \TE = \Own H R \TETwo\).
      Hence, it remains to show that \( \SAT H R {\TETwo } \).
      In particular, we need to show that \( \SATv H R {R(x)} {\Tint \nu {\form \land \nu = 0}}  \).
      Since \( R(x) \le 0 \), we have \( \models [R](x \le  0)\).
      It remains to show  \( \models [R][R(x) / \nu ] \form \), which we have from \( \SATv H R {R(x)} {\TE(x)} \).
    \item[Case \rname{R-IfFalse}]
      Similar to the case of \rname{T-IfTrue}.
    \item[Case \rname{R-Call}]
      In this case, we have
      \begin{gather*}
        e = \letexp{x}{f(y_1,\dots ,y_n)}e_2 \qquad e' = \letexp{x}{[y_1/x_1] \cdots [y_n/x_n] e_1}e_2 \\
        f \mapsto (x_1, \ldots, x_n) e_2 \in D \\
        H' = H \qquad R' = R
      \end{gather*}
      for some \( x\), \( f \), \( e_1 \) and \( e_2 \).
      Without loss of generality we may assume that \( x_1, \ldots, x_n \) are fresh.
      By inversion on the typing of \( e \), we have
      \begin{gather}
        \TE[y_1 : \sigma \ty_1] \cdots [y_n : \sigma \ty_n] \nonumber \\
        \FE(f)=\langle x_1:\ty_1, \dots ,x_n:\ty_n\rangle\rightarrow\langle x_1:\ty_1',\dots ,x_n:\ty_n'\mid \ty_x \rangle \nonumber \\
        \FE\mid\TE[y_i\update\sigma\ty_i'],x: \sigma \ty_x \p e_2:\ty\To\TE'\andalso x\notin\dom(\TE') \cup \FV(\ty)  \label{eq:lem:subject-reduction:call:exp-two-typable}
      \end{gather}
      for some \( \ty_x \), \( \ty_i\) and \( \ty'_i\) where \( \sigma=[y_1/x_1]\cdots[y_n/x_n]  \).
      Moreover, because \( \FE \p D \), we must have
      \begin{align}
        \FE \mid x_1 : \ty_1 \ldots x_n : \ty_n \p e_1 : \ty_x \To x_1 : \ty'_1, \ldots, x_n : \ty_n'  \label{eq:lem:subject-reduction:call:exp-one-typable}
      \end{align}
      by the rules \( \rname{T-Funs} \) and \( \rname{T-FunDef}\).

      We take \( \TE \) as \( \TETwo \).
      Since the preservation of the invariants are trivial, we only need to check that \( e' \) can be typed under \( \TE \).
      By repeatedly applying the substitution lemma (Lemma~\ref{lem:substitution}) to~\eqref{eq:lem:subject-reduction:call:exp-one-typable}, and then applying a standard weakening lemma, we obtain
      \begin{align*}
        \FE \mid \TE \p [y_1/x_1] \cdots [y_n / x_n] e_1 : \sigma \ty_x \To \TE[y_i\update\sigma\ty_i'].
      \end{align*}
      Applying \rname{T-Let} to this and~\eqref{eq:lem:subject-reduction:call:exp-one-typable} gives \( \FE \mid \TE \p e' : \ty \To \TE' \) as desired.
    \item[Case \rname{R-AliasDerefFail} and \rname{R-AliasAddPtrFail}]
      Trivial.
  \end{description}
\end{proof}

\subsection{Proof of Progress}
\label{sec:progress}
Finally, we prove the progress lemma (Lemma~\ref{lem:progress}) to conclude the soundness of our type system.
We first observe that any expression can be decomposed into an evaluation context and a redex.
\begin{lemma}
  \label{lem:decomposition}
  For any expression \( e \), either (i) \( e = x \) for some variable \( x \), (ii) \( e = n \) for some integer \( n \), or (iii) \( e = E[e'] \) for some evaluation context \( E \) and an expression \( e' \) where \( e' \) is one of the following forms:
  \begin{enumerate}
    \item \( \letexp{x}{n}e_0 \)
    \item \( \letexp{x}{y}e_0  \)
    \item \( \ifexp{x}{e_1}{e_2} \)
    \item \( \letexp{x}{*y}e_0\)
    \item \( x:=y;e_0 \)
    \item \( \letexp{x}{y\pplus z}e_0 \)
    \item \( \letexp{x}{\mkarray\;n}e_0 \)
    \item \( \letexp{x}{f(y_1,\dots ,y_n)}e_0 \)
    \item \( \aliasexp{x}{*y}e_0 \)
    \item \( \aliasexp{x}{y\pplus z}e_0 \)
    \item \( \assertexp{\form}; e_0 \)
  \end{enumerate}
\end{lemma}
\begin{proof}
  By a straight forward induction on the structure of \( e \).
\end{proof}

\Progress*
\begin{proof}
  The proof is by a case analysis on the shape of \( e \) according to Lemma~\ref{lem:decomposition}.
  The case where \( e = x \) or \( e = n \) is exactly the case \ref{it:lem:progress:halt}, so it remains to consider the case where \( e = E[e'] \).
  Now the proof proceeds by induction on the structure of \( E \).
  We only show the base case where \( E = \hole \) because the inductive case follows immediately from the induction hypothesis.

  Note that, by \( \vdash_D \sconfig R H e : \ty \To \TE' \), we have \( \FE \mid \TE \vdash e : \ty \To \TE' \) and \( \SAT R H {\TE} \) for some \( \TE \).
  In particular, this implies that
  \begin{align}
    &\forall x \in \dom(\TE).\ x \in \dom(R) \\
    &\text{if \( x : \Tint {\nu}{\form} \in  \TE \), then \( R(x) \in \mathbb Z \)} \label{eq:lem:progress:int}\\
    &\text{if \( x : \Tref{\idx}{\ty'} r \in \TE\), then \( R(x) = \Addr \addr \idx \) for some \( \Addr \addr \idx \) and \( \Addr \addr \idx \in \dom(H) \).} \label{eq:lem:progress:ref}
  \end{align}

  Now we proceed by a case analysis on the shape of \( e' \).
  \begin{description}
    \item[Case \( e' = \letexp{x}{n}e_0 \) and \( e'=  \letexp{x}{y}e_0 \):]
      The configuration can reduce using \rname{R-LetInt} and \rname{R-LetVar}, respectively.
    \item[Case \( e' = \ifexp{x}{e_1}{e_2} \):]
      By \eqref{eq:lem:progress:int}, we have \( R(x) \in \mathbb Z \).
      The configuration can reduce depending on the value \( R(x) \) using \rname{R-IfTrue} or \rname{R-IfFalse}.
    \item[Case \( e' = \letexp{x}{*y}e_0\):]
      By inversion on the typing of \( e' \),  \( y \) must be a reference type.
      Therefore, by~\eqref{eq:lem:progress:ref}, \( H(R(y)) \) is defined, and the configuration can reduce using \rname{R-Deref}.
    \item[Case \( e' = x:=y;e_0 \):]
      We can reduce the configuration using \rname{R-Assign} if \( R(x) \in \dom(H) \).
      The condition \( R(x) \in \dom(H) \) is met by~\eqref{eq:lem:progress:ref} and the fact that \( x \) has a reference type, which follows from the well-typedness of \( e' \).
    \item[Case \( e' = \letexp{x}{y\pplus z}e_0 \):]
      Since \( R(z) \in \mathbb Z \) by~\eqref{eq:lem:progress:int} and the typing of \( e' \), we can reduce the configuration using \rname{R-AddPtr}.
    \item[Case \( e' = \letexp{x}{\mkarray\;n}e_0 \):]
      This case is obvious as we can reduce the configuration according to \rname{R-MkArray}.
    \item[Case \( \letexp{x}{f(y_1,\dots ,y_n)}e_0 \):]
      It suffices to check that \( f \mapsto (x_1, \ldots, x_n) e'' \in D \).
      Then we can conclude this case using \rname{R-Call}.
      By inversion on the typing of \( e' \), we have \( f \in \dom(\FE) \).
      We must also have \( \FE \vdash D \) from \( \vdash_D  \sconfig R H e : \ty \To \TE' \).
      By the definition of  \( \FE \vdash D \), we have \( f \mapsto (x_1, \ldots, x_n) e'' \in D \) as desired.
    \item[Case \( \aliasexp{x}{*y}e_0 \):]
      By~\eqref{eq:lem:progress:int},~\eqref{eq:lem:progress:ref} and the fact that \( e' \) is well-typed, we know that \( R(x) \) and \( H(R(y)) \) are defined.
      Hence, we can reduce the configuration using \rname{R-AliasDeref} or \rname{R-AliasDerefFail}.
    \item[Case \( e' = \aliasexp{x}{y\pplus z}e_0 \):]
      Similar to the previous case. %
    \item[Case \( e' = \assertexp{\form}; e_0 \):]
      We can reduce the configuration using \rname{R-Assert} or \rname{R-AssertFail} according to whether \( [R] \models \form \) holds or not.
      However, by Lemma~\ref{lem:assert-fail-never-happens} we know that the configuration cannot reduce to assert failure, so the state after the reduction step must be a configuration.
  \end{description}
\end{proof}

%% file: benchmark.tex
\section{Benchmark Programs}
\label{sec:benchmark}

\lstdefinestyle{mystyleML}{
    language=caml,
    basicstyle={\footnotesize\ttfamily},
    identifierstyle={\small},
    commentstyle={\small\ttfamily \color[rgb]{0,0,1}},
    keywordstyle={\small\ttfamily \color[rgb]{0,0.5,0}},
    ndkeywordstyle={\small},
    stringstyle={\small\ttfamily \color[rgb]{0,0,1}},
    frame={tb},
    breaklines=true,
    columns=[l]{fullflexible},
    xrightmargin=0pt,
    xleftmargin=3pt,
    numberstyle={\scriptsize},
    stepnumber=1,
    numbersep=1pt,
    morecomment=[s]{[}{]},
    morekeywords={alloc, assert, alias},
    mathescape=true,
}

Figures~\ref{fig:bench-init10}--\ref{fig:bench-add} list the benchmark programs used in
the experiments described in Section~\ref{sec:experiment}.
Note that those programs were given as they are, without any annotations
of ownerships and refinement predicates. For the comparison with SeaHorn,
we prepared corresponding C programs and passed them to SeaHorn.

\begin{figure}[h]
\begin{lstlisting}[style=mystyleML] 
init(n, p) 
[ <n: int, p: int ref> -> 
  <n: int, p: int ref | int> ]
{
  if n <= 0 then {
    1
  } else {
    p := 0; let q = p + 1 in let m = n - 1 in
    let d = init(m, q) in 0
  }
}

init_assert(n, p)
[ <n: int, p: int ref> -> 
  <n: int, p: int ref | int> ]
{
  if n <= 0 then {
    1
  } else {
    let y = *p in assert(y = 0); let q = p + 1 in let m = n - 1 in 
    let d = init_assert(m, q) in 0
  }
}

{
  let p = alloc 10 in let m = 10 in
  let d1 = init(m, p) in let d2 = init_assert(m, p) in 0
}
\end{lstlisting}
\caption{\textbf{Init-10}}
\label{fig:bench-init10}
\end{figure}

\begin{figure}
\begin{lstlisting}[style=mystyleML] 
init(n, p) 
[ <n: int, p: int ref> -> 
  <n: int, p: int ref | int> ]
{
  if n <= 0 then {
    1
  } else {
    p := 0; let q = p + 1 in let m = n - 1 in
    let d = init(m, q) in 0
  }
}

init_assert(n, p)
[ <n: int, p: int ref> -> 
  <n: int, p: int ref | int> ]
{
  if n <= 0 then {
    1
  } else {
    let y = *p in assert(y = 0); let q = p + 1 in let m = n - 1 in 
    let d = init_assert(m, q) in 0
  }
}

{
  let p = alloc 1000 in let m = 1000 in
  let d1 = init(m, p) in let d2 = init_assert(m, p) in 0
}
\end{lstlisting}
\caption{\textbf{Init}}
\label{fig:bench-init}
\end{figure}

\begin{figure}
\begin{lstlisting}[style=mystyleML] 
sum(n, p) 
[ <n: int, p: int ref> -> 
  <n: int, p: int ref | int> ]
{
  if n <= 0 then {
    0
  } else {
    let x = *p in let q = p + 1 in let m = n - 1 in
    let s = sum(m, q) in let y = x + s in y
  }
}

abs(m)
[ <m: int> -> <m: int | int> ]
{
  if m >= 0 then {
    m
  } else {
    let k = -m in k
  }
}

init_x(n, x, p) 
[ <n: int, x: int, p: int ref> -> 
  <n: int, x: int, p: int ref | int> ]
{
  if n <= 0 then {
    1
  } else {
    p := x; let q = p + 1 in let m = n - 1 in
    let d = init_x(m, x, q) in 0
  }
}

{
  let p = alloc 1000 in let m = 1000 in
  let rand = _ in let z = abs(rand) in
  let d = init_x(m, z, p) in let x = sum(m, p) in
  assert(x >= 0); 0
}

\end{lstlisting}
\caption{\textbf{Sum}}
\label{fig:bench-sum}
\end{figure}

\begin{figure}
\begin{lstlisting}[style=mystyleML] 
sum_back(n, p) 
[ <n: int, p: int ref> -> 
  <n: int, p: int ref | int> ]
{
  if n <= 0 then {
    0
  } else {
    let m = n - 1 in let q = p + m in let x = *q in
    let s = sum_back(m, p) in let y = x + s in y
  }
}

abs(m)
[ <m: int> -> <m: int | int> ]
{
  if m >= 0 then {
    m
  } else {
    let k = -m in k
  }
}

init_x(n, x, p) 
[ <n: int, x: int, p: int ref> -> 
  <n: int, x: int, p: int ref | int> ]
{
  if n <= 0 then {
    1
  } else {
    p := x; let q = p + 1 in let m = n - 1 in
    let d = init_x(m, x, q) in 0
  }
}

{
  let p = alloc 1000 in let m = 1000 in
  let rand = _ in let z = abs(rand) in
  let d = init_x(m, z, p) in let x = sum_back(m, p) in
  assert(x >= 0); 0
}

\end{lstlisting}
\caption{\textbf{Sum-Back}}
\label{fig:bench-sumback}
\end{figure}

\begin{figure}
\begin{lstlisting}[style=mystyleML] 
sum_both(n, p) 
[ <n: int, p: int ref> -> 
  <n: int, p: int ref | int> ]
{
  if n <= 0 then {
    0
  } else {
    let x = *p in
    if n = 1 then {
      x
    } else {
      let k = n - 1 in let q = p + k in let x2 = *q in
      let y = x + x2 in let m = n - 2 in let p' = p + 1 in
      let s = sum(m, p') in let z = y + s in z
    } 
  }
}

abs(m)
[ <m: int> -> <m: int | int> ]
{
  if m >= 0 then {
    m
  } else {
    let k = -m in k
  }
}

init_x(n, x, p) 
[ <n: int, x: int, p: int ref> -> 
  <n: int, x: int, p: int ref | int> ]
{
  if n <= 0 then {
    1
  } else {
    p := x; let q = p + 1 in let m = n - 1 in
    let d = init_x(m, x, q) in 0
  }
}

{
  let p = alloc 1000 in let m = 1000 in
  let rand = _ in let z = abs(rand) in
  let d = init_x(m, z, p) in let x = sum_both(m, p) in
  assert(x >= 0); 0
}

\end{lstlisting}
\caption{\textbf{Sum-Both}}
\label{fig:bench-sumboth}
\end{figure}

\begin{figure}
\begin{lstlisting}[style=mystyleML] 
sum_div(n, m, p) 
[ <n: int, m: int, p: int ref> -> 
  <n: int, m: int, p: int ref | int> ]
{
  if n < m then {
    0
  } else {
    if n <= 0 then {
      0
    } else {
      let x = *p in
      if n = 1 then {
        x
      } else {
        let q = p + m in let m2 = m / 2 in let s1 = sum_div(m, m2, p) in
        let k = n - m in let k2 = k / 2 in let s2 = sum_div(k, k2, q) in
        let z = s1 + s2 in z
      } 
    }
  }
}

abs(m)
[ <m: int> -> <m: int | int> ]
{
  if m >= 0 then {
    m
  } else {
    let k = -m in k
  }
}

init_x(n, x, p) 
[ <n: int, x: int, p: int ref> -> 
  <n: int, x: int, p: int ref | int> ]
{
  if n <= 0 then {
    1
  } else {
    p := x; let q = p + 1 in let m = n - 1 in
    let d = init_x(m, x, q) in 0
  }
}

{
  let p = mkarray 1000 in let m = 1000 in 
  let rand = _ in let z = abs(rand) in let d = init_x(m, z, p) in 
  let m2 = 500 in let x = sum_div(m, m2, p) in
  assert(x >= 0); 0
}

\end{lstlisting}
\caption{\textbf{Sum-Div}}
\label{fig:bench-sumdiv}
\end{figure}

\begin{figure}
\begin{lstlisting}[style=mystyleML] 
copy_array(n, p, q)
[ <n: int, p: int ref, q: int ref> -> 
  <n: int, p: int ref, q: int ref | int> ]
{
  if n <= 0 then {
    1
  } else {
    let x = *p in q := x; let p' = p + 1 in let q' = q + 1 in
    let m = n - 1 in let d = copy_array(m, p', q') in 0
  }
}

abs(m)
[ <m: int> -> <m: int | int> ]
{
  if m >= 0 then {
    m
  } else {
    let k = -m in k
  }
}

init_x(n, x, p) 
[ <n: int, x: int, p: int ref> -> 
  <n: int, x: int, p: int ref | int> ]
{
  if n <= 0 then {
    1
  } else {
    p := x; let q = p + 1 in let m = n - 1 in
    let d = init_x(m, x, q) in 0
  }
}

copy_assert(n, p)
[ <n: int, p: int ref> -> 
  <n: int, p: int ref | int> ]
{
  if n <= 0 then {
    1
  } else {
    let y = *p in assert(y >= 0); let q = p + 1 in 
    let m = n - 1 in let d = copy_assert(m, q) in 0
  }
}

{
  let p = alloc 1000 in let q = alloc 1000 in let m = 1000 in
  let rand = _ in let z = abs(rand) in let d1 = init_x(m, z, p) in 
  let d2 = copy_array(m, p, q) in let d3 = copy_assert(m, q) in 0
}

\end{lstlisting}
\caption{\textbf{Copy-Array}}
\label{fig:bench-copy}
\end{figure}

\begin{figure}
\begin{lstlisting}[style=mystyleML] 
add_array(n, p, q, r)
[ <n: int, p: int ref, q: int ref, r: int ref> -> 
  <n: int, p: int ref, q: int ref, r: int ref | int> ]
{
  if n <= 0 then {
    1
  } else {
    let x = *p in let y = *q in let z = x + y in r := z;
    let p' = p + 1 in let q' = q + 1 in let r' = r + 1 in
    let m = n - 1 in let d = add_array(m, p', q', r') in 0
  }
}

abs(m)
[ <m: int> -> <m: int | int> ]
{
  if m >= 0 then {
    m
  } else {
    let k = -m in k
  }
}

init_x(n, x, p) 
[ <n: int, x: int, p: int ref> -> 
  <n: int, x: int, p: int ref | int> ]
{
  if n <= 0 then {
    1
  } else {
    p := x; let q = p + 1 in let m = n - 1 in
    let d = init_x(m, x, q) in 0
  }
}

add_assert(n, p)
[ <n: int, p: int ref> -> 
  <n: int, p: int ref | int> ]
{
  if n <= 0 then {
    1
  } else {
    let y = *p in assert(y >= 0); let q = p + 1 in 
    let m = n - 1 in let d = add_assert(m, q) in 0
  }
}

{
  let p = alloc 1000 in let q = alloc 1000 in let r = alloc 1000 in
  let m = 1000 in let rand = _ in let z = abs(rand) in 
  let d1 = init_x(m, z, p) in let d2 = init_x(m, z, q) in 
  let d3 = add_array(m, p, q, r) in let d4 = add_assert(m, r) in 0
}

\end{lstlisting}
\caption{\textbf{Add-Array}}
\label{fig:bench-add}
\end{figure}